\documentclass[a4paper,UKenglish,cleveref, autoref, thm-restate,
authorcolumns]{lipics-v2021}
\usepackage{arxiv-lipics-v2021}

\hideLIPIcs  %
\nolinenumbers
\usepackage{mathtools}
\usepackage{tikz}
\usetikzlibrary{automata,snakes,positioning,arrows.meta,arrows,
decorations.pathreplacing,calligraphy,petri,shapes.geometric}
\usepackage[bibliography=common,forwardlinking=yes]{apxproof}
\usepackage{marginnote}

\makeatletter
\newcommand{\latestappendixlabel}{}
\newcommand{\appendixlabel}[1]{
  \renewcommand\latestappendixlabel{#1}%
}
\makeatother

\title{Representing One Letter Weighted Automata over the Tropical Semiring} %

\author{Shaull Almagor}{Technion, Haifa, Israel}{shaull@technion.ac.il}{https://orcid.org/0000-0001-9021-1175}{Supported by the ISRAEL SCIENCE FOUNDATION (grant No. 989/22)}

\author{Ismaël Jecker}{Université Marie et Louis Pasteur, 
CNRS, institut FEMTO-ST, F-25000 Besançon, 
France}{ismael.jecker@gmail.com}{https://orcid.org/0000-0002-6527-4470}{Supported by the Agence Nationale de la Recherche (ANR) grant ANR-25-CE48-2447 FAVOR.}

\author{Filip Mazowiecki}{University of Warsaw, Poland}{f.mazowiecki@mimuw.edu.pl}{https://orcid.org/0000-0002-4535-6508}{Supported by the Polish National Science Centre SONATA BIS-12 grant number 2022/46/E/ST6/00230.}%

\author{Łukasz Orlikowski}{University of Warsaw, Poland}{l.orlikowski@uw.edu.pl}{https://orcid.org/0009-0001-4727-2068}{Supported by the ERC grant INFSYS, agreement no. 950398.}

\author{David Purser}{University of Liverpool, UK}{D.Purser@liverpool.ac.uk}{https://orcid.org/0000-0003-0394-1634}{}

\author{Henry Sinclair-Banks}{University of Warsaw, Poland \and Max Planck Institute for Software Systems (MPI-SWS), Kaiserslautern, Germany \and \url{http://henry.sinclair-banks.com}}{hsb@mpi-sws.org}{https://orcid.org/0000-0003-1653-4069}{Supported by the ERC grant INFSYS, agreement no. 950398.}%

\authorrunning{S. Almagor, I. Jecker, F. Mazowiecki, \L{}. Orlikowski, D. Purser, and H. Sinclair-Banks} %

\Copyright{Shaull Almagor, Ismaël Jecker, Filip Mazowiecki, Łukasz Orlikowski, David Purser, and Henry Sinclair-Banks} %



\category{} %

\relatedversion{Published version \url{https://doi.org/10.4230/LIPIcs.CONCUR.2026.10}} %

\acknowledgements{We acknowledge the \href{https://samsa2025.github.io/rock}{Weighted Automata Rock workshop} funded by IDUB Thematic Research Programmes (POB III) of the University of Warsaw.}%

\EventEditors{Ana Sokolova and Patrick Totzke}
\EventNoEds{2}
\EventLongTitle{37th International Conference on Concurrency Theory (CONCUR 2026)}
\EventShortTitle{CONCUR 2026}
\EventAcronym{CONCUR}
\EventYear{2026}
\EventDate{September 1--4, 2026}
\EventLocation{Liverpool, UK}
\EventLogo{}
\SeriesVolume{391}
\ArticleNo{45}

\usepackage{todonotes}

\newcommand{\bbN}{\mathbb{N}}
\newcommand{\NN}{\bbN}
\newcommand{\bbZ}{\mathbb{Z}}
\newcommand{\ZZ}{\bbZ}
\newcommand{\bbQ}{\mathbb{Q}}

\newcommand{\bbNinf}{\mathbb{N}_{\infty}}
\newcommand{\bbZinf}{\mathbb{Z}_{\infty}}
\newcommand{\set}[1]{\left\{#1\right\}}
\newcommand{\sset}{\subseteq}
\newcommand{\tup}[1]{\langle #1 \rangle}

\newcommand{\abs}[1]{\lvert#1\rvert}

\newcommand{\norm}[1]{\|#1\|}

\newcommand{\cA}{\mathcal{A}}
\newcommand{\cB}{\mathcal{B}}
\newcommand{\cC}{\mathcal{C}}
\newcommand{\cD}{\mathcal{D}}

\newcommand{\cN}{\mathcal{N}}

\newcommand{\init}{\texttt{init}}
\newcommand{\fin}{\texttt{fin}}
\newcommand{\length}{\mathsf{len}}
\newcommand{\minweight}{\mathsf{mwt}}

\newcommand{\runsto}[1]{\xrightarrow{#1}}

\newcommand{\upd}{\texttt{upd}}
\newcommand{\weight}{\mathsf{wt}}
\newcommand{\opt}{\mathsf{opt}}

\newcommand{\dom}{\textsf{dom}}
\newcommand{\class}[1]{\ifmmode\mathsf{#1}\else\textsf{#1}\fi}
\usepackage{xspace}
\newcommand{\coNP}{\class{coNP}\xspace}
\newcommand{\NP}{\class{NP}\xspace}

\newcommand{\wone}{$\class{W_{1}}$}
\newcommand{\cowone}{$\class{coW_{1}}$}

\tikzstyle{state} = [circle, draw, black, line width = 0.3mm, inner sep = 1mm]
\tikzstyle{transition} = [-{Stealth[width=1.5mm, length=2mm]}, black, line width = 0.3mm]

\newtheoremrep{prop}[proposition]{Proposition}
\newtheoremrep{lem}[lemma]{Lemma}
\newtheoremrep{theo}[theorem]{Theorem}
\newtheoremrep{coro}[corollary]{Corollary}
\theoremstyle{claimstyle}
\newtheoremrep{cla}[claim]{Claim} 
\begin{document}

\maketitle

\begin{abstract}
We consider weighted automata over the tropical semiring $\bbZinf(\min, +)$. Recently, it was shown that determinisation is decidable; in this paper we focus on the complexity when the alphabet is unary. In 2001, Lombardy showed this problem is decidable, a close inspection of his proof yields a \coNP upper bound on the complexity. Earlier Gaubert showed that every weighted automaton in this setting can be effectively turned into an equivalent union of deterministic weighted automata. We prove Gaubert's result efficiently, presenting it as a generalisation of Chrobak's normal form for unary NFA. In particular, we prove that the equivalent union of deterministic weighted automata can be represented by a weighted automaton of quadratic size in the size of the original one, and this representation can be computed in polynomial time. Building on this, we show that determinisation, and even register minimisation (which generalises determinisation), is \coNP-complete. We complete the paper with observations that the boundedness problem is also \coNP-complete by reductions with determinisation.
Lastly, we provide evidence that all of these problems are not FPT (by proving \cowone-hardness) when parametrised by the number of deterministic automata in the union.
\end{abstract}

\section{Introduction}
\label{sec:introduction}

Weighted automata over the tropical semiring $(\bbZ\cup\{\infty\}, \min, +)$ are a quantitative extension of finite automata by labelling transitions with integers. 
Instead of recognising input words, they assign them integer weights. 
To every run we assign the sum of the weights on its transition, and the output for a given input word aggregates the weigths of accepting runs using the minimum function. Weighted automata over the tropical semiring play a central role in verification, see for example the surveys~\cite{LombardyS06,Daviaud20}. 
For example, this model played a crucial role in Hashiguchi's results proving decidability of the star height problem~\cite{Hashiguchi82,Hashiguchi82b}.

The power of weighted automata over the tropical semiring makes the model convenient for expressing many properties, unfortunately at the cost that most problems are undecidable~\cite{AlmagorBK22}. 
In contrast to Boolean finite automata, weighted automata cannot be always determinised. 
Furthermore, decision problems for subclasses, like deterministic or unambiguous weighted automata, become significantly less complex and, in particular, decidable~\cite{LombardyS06,Daviaud20}. 
This naturally raises the question whether an input automaton admits an equivalent deterministic weighted automaton; as a decision problem it is called the determinisation problem. A recent breakthrough established decidability of determinisation for tropical weighted automata~\cite{almagor2026determinization}, after decades from its formulation in the 1990s by Mohri~\cite{mohri1994compact,mohri1997finite} (and already implicit in earlier work~\cite{choffrut1977caracterisation}).

In this paper, we focus on weighted automata over the tropical semiring when the alphabet is unary.
In this setting decidability of determinisation was established by Lombardy~\cite{lombardy2001sequentialization}. While not stated explicitly, a $\coNP$  complexity bound can be inferred from~\cite{lombardy2001sequentialization}, matching our upper bound.
Indeed, Lombardy essentially reduces the problem to the setting when the input automaton is assumed to be unambiguous. 
Then, non-determinisability can be witnessed by word demonstrating a cycle in a certain automaton.
More recently, determinisation for unambiguous automata over the unary alphabet was shown to be decidable in linear time~\cite[Theorem 1.4]{abs-2501-14725}.

\subparagraph*{Our contributions.}
We revisit the unary alphabet setting with the aim of obtaining a precise structural and complexity-theoretic understanding of determinisation and related problems. Our approach is based on a refined analysis of runs using cycles of minimal average weight. Intuitively, any sufficiently long run of minimal weight must predominantly iterate such cycles. We show that this intuition can be turned into an effective normal form that captures all behaviours in a compact and algorithmically tractable way.
Gaubert~\cite{gaubert2005rational} shows that every unary weighted automaton can be transformed into an equivalent finite union of deterministic weighted automata. 
We improve this reduction showing that, with respect to the size of the original automaton, (i) there are linearly many automata in the union, (ii) each automaton in the union has quadratic size, and (iii) the transformation can be computed in polynomial time. 
This result is a weighted generalisation of Chrobak’s normal form for unary (Boolean) automata~\cite{Chrobak86,To09}, and yields a concise ``lasso-shape'' representation governed by cycles of minimal average weight.

We also consider the register minimisation problem, which is known to strictly generalise determinisation~\cite{alur2013regular,AlurR13}. Using our Chrobak-like decomposition, we obtain a \coNP decision procedure both for register minimisation and determinisation problems. The proof is an analysis of the interactions between the deterministic components and identifying cycles that are asymptotically redundant. 
Note that over unrestricted alphabets, the register minimisation problem has recently been shown undecidable~\cite{almagor2025unambiguisability}, in contrast to the aforementioned decidability of determinisation~\cite{almagor2026determinization}.

Another important problem is the boundedness problem, studied in Hashiguchi's work on star-height~\cite{Hashiguchi82}. We show that determinisability coincides with boundedness. This link enables us to transfer techniques and results between the two settings. 
In particular, we obtain a \coNP upper bound for boundedness.

We complete these results with lower bounds proving tight complexity bounds. We show that determinisation, register minimisation, and boundedness are all \coNP-complete for unary weighted automata.
Given our Chrobak-like result, it is natural to ask for FPT algorithms, when the input is explicitly a finite union of deterministic weighted automata. 
We prove that all three problems are \cowone-hard when parameterised by the number of deterministic components in the decomposition, relying on Fernau, Hoffmann, and Wehar's work on the unary DFA intersection problem~\cite{intersection_w1,Wehar16}. 
This indicates that determinisation, register minimisation, and boundedness remain intractable even under strong structural restrictions.

\subparagraph*{Related work.} Weighted automata are also studied over fields, in particular over the rationals ($\bbQ,+,\times$)~\cite{Schutzenberger61b}, encompassing probabilistic automata as a special case~\cite{GimbertO10}. The determinisation and register minimisation problems have been shown to be decidable in the rational semiring~\cite{BellS23}, with the exact complexity still unresolved~\cite{BenaliouaLR24,JeckerMP24} and polynomial time over the unary alphabet~\cite{Kostolanyi22}.

In general, one can also consider the ambiguity hierarchy, which is a fine-grained analysis of classes between deterministic and unrestricted non-deterministic automata. 
Most notable partial results on determinisability are for: unambiguous~\cite{mohri1997finite}, finitely ambiguous~\cite{klimann2004deciding}, and polynomially ambiguous automata~\cite{kirsten2009deciding}. 
One can also generalise the determinisation problem by asking whether the input automaton admits an alternative representation that is unambiguous, finite-ambiguous, or polynomial-ambiguous~\cite{almagor2025unambiguisability,abs-2410-03444}. However, these problems trivialise in the case of unary alphabet over the tropical semiring. 
Recall that Lombardy showed that all unary weighted automata can be converted into unambiguous weighted automaton~\cite{lombardy2001sequentialization}.

Taking the algebraic view of automata, a unary weighted automaton can be viewed as a matrix $A\in \bbZ^{d\times d}$ over the tropical semiring, where the \emph{orbit} $\{A^n\mid n\in \bbN\}$ captures the behaviour of the automaton. Under this view, several works~\cite{sergeev2012csr,nachtigall1997powers,hartmann1999transience} have identified that the mean-value of ``cycles'' can be used to effectively reason about the orbit. Our work also contains several such observations, tailored to our setting. However, we stay in the view of state machines.

\section{Preliminaries}
\label{sec:preliminaries}

For $k\in \bbN=\{0,1,\ldots\}$, we denote $[k] \coloneqq\{1,\ldots,k\}$. 
We denote by $\bbNinf$ and $\bbZinf$ the sets $\bbN\cup \{\infty\}$ and $\bbZ\cup \{\infty\}$, respectively. We extend the addition and $\min$ operations to $\infty$ in the natural way: $a+\infty=\infty$ and $\min\{a,\infty\}=a$, for all $a\in \bbZinf$. By $\arg\min\{f(x) : x\in A\}$ we mean the set of elements in $A$ for which $f(x)$ is minimal for some function $f$ and set $A$.

\subparagraph*{Weighted automata.}
A \emph{Unary $(\min,+)$ Weighted Finite Automaton} (unary WFA for short) is a tuple $\cA= \tup{Q, Q_0, \Delta, F}$, where
    $Q$ is a finite set of \emph{states},
    $Q_0\subseteq Q$ are the \emph{initial states},
    $\Delta\subseteq Q\times \bbZinf\times Q$ is a \emph{transition relation} such that for every $p,q\in Q$ there exists exactly\footnote{This is without loss of generality: if there are two transitions with different weights, the higher weight can always be ignored in the $(\min,+)$ semantics. Missing transitions can be introduced with weight $\infty$.} one weight $c\in \bbZinf$ such that $(p,c,q)\in \Delta$, and
    $F\subseteq Q$ is a set of \emph{accepting states}.
If $|Q_0|=1$ and, for every $p\in Q$, there exists at most one transition $(p,c,q)$ with $c\neq \infty$, then $\cA$ is called \emph{deterministic}. 
We denote by $\norm{\cA}$ the maximal absolute value of any finite weight in $\Delta$.

\begin{remark}[Initial and final weights and states]
    \label{rmk: initial and final weights}
    Weighted automata are often defined with initial and final weights, i.e., 
    $Q_0$ and $F$ are replaced by initial and final vectors $\init,\fin\in \bbZinf^Q$. The weight of a run then also includes the initial weight and final weight. Our choice is with only a small loss of generality, the expressivity of the two models differ only on the empty word, which has no effect on our results.
\end{remark}
\begin{remark}
    Throughout this paper we use the $(\min,+)$ semiring. Since we work over $\bbZ$, our results also apply to $(\max,+)$ by symmetry.
\end{remark}

\subparagraph*{Runs.}
A \emph{run} of $\cA$ is a sequence of transitions $\rho=(t_1,t_2,\ldots,t_m)$ where $t_i=(p_i,c_i,q_i)$ such that $q_i=p_{i+1}$ for every $1\le i<m$, and $c_i <\infty$ for every $1\le i \le m$; we write $\rho:p_1\runsto{m}q_m$ to denote this run.
We say that $\rho$ is \emph{of length $m$ from $p_1$ to $q_m$}; we also use the notation $\length(\rho) \coloneqq m$ to denote the length of a run. 
For an infix $w[i,j]$, we denote the corresponding infix of $\rho$ by $\rho[i,j] \coloneqq (t_i,\ldots,t_j)$.
The \emph{weight} of the run $\rho$ is $\weight(\rho)=\sum_{i=1}^m c_i$. A run is \emph{accepting} if $p_1\in Q_0$ and $q_m\in F$.

For $n\in \bbN$, we abuse the name of the unary WFA as the function it describes, and denote by ${\cA}(n)$ the weight assigned by $\cA$ to $n$, which is the minimal weight of an accepting run of $\cA$ of length $n$. For convenience, we introduce some auxiliary notations.
For $n\in \bbN$ and sets of states $Q_1,Q_2\subseteq Q$, we define (with the convention $\min \emptyset=\infty$)
\[\minweight_\cA(Q_1\runsto{n} Q_2) \coloneqq \min\{\weight(\rho)\mid \exists q_1\in Q_1,q_2\in Q_2,\ \rho:q_1\runsto{n}q_2\}.\]
If $Q_1$ or $Q_2$ are singleton, then we denote them by a single state (e.g., $\minweight_\cA(P\runsto{n} q)$ for some set $P\subseteq Q$ and state $q \in Q$).
Accordingly, we can define
$\cA(n)=\minweight_\cA(Q_0\runsto{n} F)$.
If there are no accepting runs on $n$, then $\cA(n)=\infty$.
The unary WFA then defines the function $\cA:\bbN\to \bbZinf$. We say the domain of $\cA$, $\dom(\cA)$, is the set $\{n\in\mathbb{N} \mid \cA(n) < \infty\}$

We write $p\runsto{n}q$ when there exists some run $\rho$ 
such that $\rho:p\runsto{n}q$.
We lift this notation to concatenations of runs, e.g., $\rho: p\runsto{n}q\runsto{m}r$ means that $\rho$ is a run on $n+m$ from $p$ to $r$ that reaches $q$ after the prefix $n$. We also incorporate this to $\minweight$ by writing $\minweight(q\runsto{n}p\runsto{m}r)$ to refer to the minimal weight of a run $\rho:q\runsto{n}p\runsto{m}r$.

\subparagraph*{Average weight.} 
Given a run $\rho$ (typically a cycle in what follows), the \emph{average weight} of $\rho$ is $\weight(\rho)/\length(\rho)$, i.e.,\ the average effect of each step in the $\rho$. When $\rho$ is a cycle, we sometimes call the average weight the \emph{slope}. 

\subparagraph*{Trim assumption.} 
A WFA is \emph{trim} if every state is reachable from $Q_0$ by some run and can reach an accepting state with some run. Note that states that do not satisfy this can be found in linear time (by using a simple graph search algorithm), and can be removed from the WFA without changing the weight of any accepted word. We assume that all WFAs are trim throughout.

\subparagraph*{$k$-Width WFAs.}
Consider a WFA $\cA$. The \emph{width} of $\cA$ is the maximal number of states that are simultaneously reachable in $\cA$.

\section{Efficiently Computing Normal Forms for Unary WFA}
\label{sec:chrobak}

\appendixlabel{appendix:chrobak}
\begin{toappendix}
\label{appendix:chrobak}
\end{toappendix}

A unary non-deterministic finite automaton is said to be in \emph{Chrobak normal form} if it consists of a path $q_1\to q_2 \to\dots\to q_m$, and $k$ simple cycles $\gamma_1,\dots, \gamma_k$, with $q_1$ initial, and $q_m$ non-deterministically branching into each of the $k$ cycles, as depicted in \cref{fig:chrobak}. 
The key result is that every $n$-state unary NFA can be converted, in polynomial time, into a unary NFA in Chrobak normal form with $m, k \le O(n^2)$ and for each cycle having length at most $n$~\cite{Chrobak86,Martinez02,To09}. 

\begin{figure}[t]
    \centering
\resizebox{0.6\textwidth}{!}{%
\begin{tikzpicture}

    \node[state, label = below:$q_1$] (q1) at (2,0) {};
    \node[state, label = below:$q_2$] (q2) at (3.5,0){};
    \node (q3) at (5,0) {$\cdots$};
    \node[state, label = below:$q_m$] (q5) at (6.5,0) {};

    \node[state] (q6) at (8,0+1.3) {};
    \node[state] (q7) at (9.5,0.55+1.3) {};
    \node[rotate=-10] (q8) at (10.75,0.55+1.3) {$\cdots$};
    \node[rotate=10] (q9) at (10.75,-0.55+1.3) {$\cdots$};
    \node[state] (q10) at (9.5,-0.55+1.3) {};

    \node[state] (p6) at (8,0-1.3) {};
    \node[state] (p7) at (9.5,0.55-1.3) {};
    \node[rotate=-10] (p8) at (10.75,0.55-1.3) {$\cdots$};
    \node[rotate=10] (p9) at (10.75,-0.55-1.3) {$\cdots$};
    \node[state] (p10) at (9.5,-0.55-1.3) {};

    \node at (10,0+1.3) {$\gamma_1$}; 
    \node at (10,0-1.3) {$\gamma_k$}; 

    \node[rotate=9] (r6) at (8,0+0.4) {$\cdots$};
    \node[rotate=-9] (l6) at (8,0-0.4) {$\cdots$};

    \node at (10,0) {$\vdots$}; 

    \draw[transition] (2-1,0) -- (q1);

    \draw[transition] (q1) -- node[above, scale=0.7]{} (q2);
    \draw[transition] (q2) -- node[above, scale=0.7]{} (q3);
    \draw[transition] (q3) -- (q5);
    \draw[transition] (q5) -- node[above, scale=0.7]{} (q6);
    \draw[transition] (q5) -- node[above, scale=0.7]{} (p6);
    \draw[transition] (q5) -- node[above, scale=0.7]{} (r6);
    \draw[transition] (q5) -- node[above, scale=0.7]{} (l6);

    \draw[transition] (q6) edge[out = 30, in = 210, bend left = 15] node[above, scale=0.6, sloped] {} (q7);
    \draw[transition] (q7) edge[out = 10, in = 180, bend left = 10] (q8);
    \draw[transition] (q9) edge[out = 10, in = 180, bend left = 10] (q10);
    \draw[transition] (q10) edge[out = 30, in = 210, bend left = 15] node[below, scale=0.6, sloped] {} (q6);

    \draw[transition] (p6) edge[out = 30, in = 210, bend left = 15] node[above, scale=0.6, sloped] {} (p7);
    \draw[transition] (p7) edge[out = 10, in = 180, bend left = 10] (p8);
    \draw[transition] (p9) edge[out = 10, in = 180, bend left = 10] (p10);
    \draw[transition] (p10) edge[out = 30, in = 210, bend left = 15] node[below, scale=0.6, sloped] {} (p6);
\end{tikzpicture}
}

     \caption{
        The structure of an automaton (without any accepting states) in Chrobak normal form.
    }
    \label{fig:chrobak}
\end{figure}
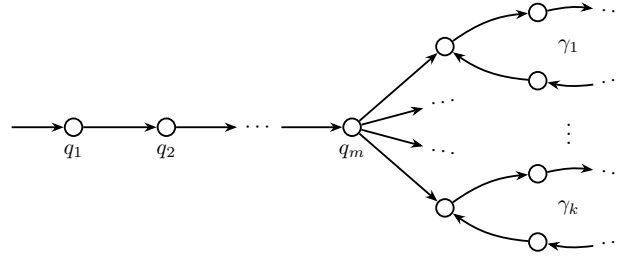

In this section, we consider two forms for unary WFAs: unions of deterministic unary WFAs and unary WFAs in Chrobak normal form.
It is known that non-deterministic unary WFA can be represented as a union of deterministic unary WFAs, which are necessarily individually lasso-shaped automata~\cite{gaubert2005rational,lombardy2001sequentialization}. 
Of course, one can also unroll the cycles so that the initial simple path segments have the same length, and then merge these initial paths (by suitably adapting the weights on the transitions of the initial path and branching transitions), to yield a single unary WFA whose underlying unary DFA is in Chrobak normal form.  
However, these results are existential and do not entail that automata in this form can be small, nor found efficiently.

For the remainder of this section, we fix our attention on a trim unary WFA $\cA = \tup{Q, Q_0, \Delta, F}$.
In~\cref{subsec:deterministic-union}, we will first prove that $\cA$ can be efficiently converted into a small union of quadratic size deterministic unary WFA (\cref{thm:deterministic-union}).
After, in~\cref{sec:cnf}, we will strengthen our result to prove that one can efficiently compute an equivalent quadratic size Chrobak normal form for $\cA$ in polynomial time (\cref{cor:normal-form}).

\subsection{Small Unions of Small Deterministic Unary WFAs}
\label{subsec:deterministic-union}

\begin{theorem} \label{thm:deterministic-union}
    There exist deterministic unary WFAs $\cB_1, \ldots, \cB_k$ such that 
    \begin{enumerate}[(1)]
        \item \label{itm:union-correctness}
        for every $n$, $\cA(n) = \min_{1 \leq i \leq k} \set{\cB_i(n)}$;
        \item \label{itm:union-cycle-count}
        $k \leq \abs{Q}$; and
        \item for every $1 \leq i \leq k$:
        \begin{enumerate}[(i)]
            \item \label{itm:union-state-count}
            $\cB_i$ consists of at most $2\abs{Q}^2 + \abs{Q}$ states and transitions, 
            \item \label{itm:union-updates}
            $\norm{\cB_i} \leq 8\abs{Q}^2\norm{\cA}$, and
            \item \label{itm:union-poly-time}
            $\cB_i$ can be computed from $\cA$ in polynomial time.
        \end{enumerate}
    \end{enumerate}
\end{theorem}

In order to prove~\cref{thm:deterministic-union}, we will identify a collection of ``representative'' cycles. 
Roughly speaking, for any $n \in \NN$, any minimal-weight run of length $n$ will mostly use cycles of minimal \emph{average} weight to get to the desired length $n$.
This is because any run of length $n$ that repeatedly uses cycles with greater weight must end up having greater weight than a run that instead uses minimal average weight cycles.
In some situations, cycles with non-minimal average weight must be taken; we will argue that such cycles will only be used a small number of times.
We formalise this in~\cref{lem:linear-form-runs}.
With this, our approach to prove~\cref{thm:deterministic-union} is to, more or less, categorise runs based on the lowest average weight short cycle used.

We call a cycle \emph{short} if it consists of at most $\abs{Q}$ transitions. 
Note that all simple cycles are short. 
For every state $q \in Q$, we \emph{distinguish} $\gamma_q$ to be a short $q$-cycle that has minimal average weight among all short $q$-cycles.
Recall that the average weight of a cycle is the ratio of its weight and length.
Note that there may be multiple such cycles for a given state $q$, we can arbitrarily distinguish any one of these cycles.

\begin{lemma}\label{lem:linear-form-runs}
    Let $n \in \NN$ and $s, t \in Q$ such that $\minweight(s\runsto{n}t)<\infty$.
    There exists a minimal weight run $\rho$ from $s$ to $t$ of length $n$ with the following properties:
    \begin{enumerate}[(1)]
        \item there is a decomposition $\rho = \tau \gamma_q^x \tau'$, for some distinguished cycle $\gamma_q$ and $x \in \NN$; and
        \item $\length(\tau) + \length(\tau') < 2\abs{Q}^2$.
    \end{enumerate}
\end{lemma}

We remark that~\cref{lem:linear-form-runs} does not specify that the run of length $n$ starts from an initial state, or leads to an accepting state.
Though, it immediately follows that if $\cA(n) \neq \infty$, then there exists an accepting run of length $n$ and weight $\cA(n)$ with properties (1) and (2) as specified by~\cref{lem:linear-form-runs}.

To prove~\cref{lem:linear-form-runs}, we will perform a cycle replacement argument on an arbitrary minimal weight run of length $n$  to obtain $\rho$.
Some of the following ideas appear in Rackoff's rearrangement and replacement arguments used to analyse the length of runs that witness unboundedness in vector addition systems~{\cite[Lemma 4.5]{Rackoff78}}. 
Similar arguments have also been used in the development of Chrobak normal form for finite automata, in particular in Lin's proof~\cite{To09} which makes a subtle correction to Chrobak's~\cite{Chrobak86}. 
However, we must take care to ensure the weights are correctly maintained.
The main technical contribution is to show that this can still be completed efficiently.
 
\begin{proof}[Proof of~\cref{lem:linear-form-runs}]
    Initially, let $\pi$ be an arbitrary run from $s$ to $t$ of length $n$ with minimal weight $w \in \ZZ$. 
    Let $Q' \sset Q$ be the subset of states that are visited by $\pi$; note that $Q' \supseteq \set{s,t}$.
    First, for every state $q \in Q'\setminus\set{s,t}$, that is visited by $\pi$, we ``mark'' exactly one instance\footnote{To be clear, some transitions may occur multiple times in $\pi$, if one is marked then only one instance of said transition is marked. We also note that, for a particular state $q \in Q'\setminus\set{s,t}$, it does not matter which incoming transition is marked.} of a transition in $\pi$ that leads to $q$.
    We also mark the final transition (which leads to $t$) and note that we need not mark any transition to $s$.
    In total, there are at most $\abs{Q}-1$ many marked transitions.
    
    \begin{clarep}\label{clm:skeleton-length}
        If $\length(\pi) \geq \abs{Q}^2$, then there is a short cycle without any marked transitions in $\pi$. 
    \end{clarep}
    \begin{claimproof}
        Consider the first $\abs{Q}^2$ transitions of $\pi$.
        In fact, consider these transitions as $\abs{Q}$ blocks of $\abs{Q}$ transitions. 
        Thanks to the pigeonhole principle, observe that every block must contain at least one short cycle.
        Now, since there are at most $\abs{Q}-1$ many marked transitions, one of these blocks consists only of transitions that are not marked.
        Thus, inside this block, there is a short cycle that does not contain any marked transitions.
    \end{claimproof}

    One by one, short cycles that do not contain marked transitions are removed from $\pi$ and placed in $C$, a multiset that is initially empty.
    The order in which short cycles are removed does not matter but, thanks to~\cref{clm:skeleton-length}, we can always identify and remove a short cycle from $\pi$ until the length of $\pi$ is less than $\abs{Q}^2$.

    Now, let $q \in Q'$ be the state visited by $\pi$ that has the lowest average weight distinguished cycle $\gamma_q$.
    Formally, 
    \begin{equation*}
        q \coloneqq \arg\min_{p\in Q'}\set{\frac{\weight(\gamma_p)}{\length(\gamma_p)}}.
    \end{equation*}
    We will now repeatedly replace cycles from $C$ with several iterations of $\gamma_q$.
    We must do this in a way that preserves the length. 
    For this we use the following combinatorial claim with $m = \abs{Q}$, $\ell = \length(\gamma_q)$, and $S$ is the multiset of all lengths of all cycles in $C$.    

    \begin{clarep} \label{clm:trading}
        Let $m, \ell \in \NN$ such that $\ell \leq m$ and let $S$ be a multiset of naturals such that, for every $s \in S$, $1 \leq s \leq m$.
        If $\sum_{s \in S} s \geq m^2$, then there exist a multiset $S' \sset S$ and $e \in \NN$ such that $\sum_{s \in S'} s = e\ell$.
    \end{clarep}
    \begin{claimproof}
        First, since $\sum_{s \in S} s \geq m^2$ and, for every $s \in S$, $s \leq m$, it is true that $\abs{S} \geq m$.
        Arbitrarily enumerate $m$ elements of $S$: $s_1, \ldots, s_m$.
        From this enumeration, we define the partial sums $t_1 = s_1$, $t_2 = s_1 + s_2$, $\ldots$, $t_m = s_1 + s_2 + \ldots + s_m$.
        There are two cases to consider.
        
        Case 1: there exists $1 \leq i \leq m$ such that $t_i \equiv 0 \bmod \ell$. Clearly, in this case, $e = \frac{t_i}{\ell}$ and $S' = \set{s_1, \ldots, s_i}$ satisfy this claim.
        
        Case 2: there does not exist $1 \leq i \leq m$ such that $t_i \equiv 0 \bmod \ell$. 
        By the pigeonhole principle, there exist $1 \leq i < j \leq m$ such that $t_i \equiv t_j \bmod \ell$. 
        This means that $t_j - t_i \equiv 0 \bmod \ell$.
        Thus, in this case $e = \frac{t_j - t_i}{\ell}$ and $S' = \set{s_{i+1}, \ldots, s_j}$ satisfy this claim.
    \end{claimproof}

    If the sum of all lengths of cycles in $C$ is at least $\abs{Q}^2$, then we can identify a collection of short cycles $\sigma_1, \ldots, \sigma_i$ such that $\length(\sigma_1) + \ldots + \length(\sigma_i) = e\cdot\length(\gamma_q)$, for some $e \in \NN$.
    This allows us to remove $\sigma_1, \ldots, \sigma_i$ to later be replaced with $e$ iterations of $\gamma_q$.
    Soon, we will argue that we can reconstruct a run of length $n$ of the form $\tau \gamma_q^x \tau'$.
    Recall that $\pi$ was assumed to have minimal weight; we therefore know that $\weight(\sigma_1) + \ldots + \weight(\sigma_i) = e\cdot \weight(\gamma_q)$.
    
    The above cycle replacement is repeated until the sum of the lengths of cycles in $C$ that remain is less than $\abs{Q}^2$.
    Finally, we will argue that a run from $s$ to $t$ of length $n$ can be reconstructed. 
    Recall that for every state in $Q' \setminus \set{s}$, there is a marked transition in $\pi$ that was not removed during the cycle removal process. 
    This means for every state $q \in Q'\setminus\set{s}$, $q$-cycles can be reinserted into $\pi$: specifically a $q$-cycle can be taken immediately after the marked transition leading to $q$.
    Furthermore, $s$-cycles can be inserted at the very beginning of the run.
    It is also important to note that since the final transition leading to $t$ was marked, we know that $\pi$ is (still) a run from $s$ to $t$. 
    Let $\pi'$ be the run that is obtained by adding all short cycles left in $C$ after replacement (specifically, $\pi'$ does not contain the replacement $x$ iterations of $\gamma_q$). 
    Since the combined length of cycles left in $C$ was at most $\abs{Q}^2$ and $\length(\pi) < \abs{Q}^2$; we deduce that $\length(\pi') < 2\abs{Q}^2$.
    
    What remains are the $x \coloneqq \frac{n - \length(\pi')}{\length(\gamma_q)}$ replacement iterations of $\gamma_q$.
    In just the same way, these cycle iterations can be inserted after the marked transition for $q$.
    Specifically, let $\tau$ be the prefix of $\pi'$ before and including the marked transition for $q$, and let $\tau'$ be the suffix of $\pi'$ after the marked transition for $q$.
    Then $\rho = \tau \gamma_q^x \tau'$ is a minimal weight run from $s$ to $t$ of length $n$.
    Clearly $\length(\tau)+\length(\tau') = \length(\pi') < 2\abs{Q}^2$.
\end{proof}

Now that we have this standard ``linear form'' for minimal weight runs, we can construct a collection of deterministic unary WFAs $\cB_1, \ldots, \cB_k$ such that, for every $n \in \NN$, $\cA(n) = \min_{1 \leq i \leq k} \set{\cB_i(n)}$.
In fact, we will construct one deterministic unary WFA $\cB_q$, for every $q \in Q$ for which there exists a (short) cycle at $q$ in $\cA$.
Recall that deterministic unary WFAs have a basic lasso-shaped structure: there is a simple path leading to a single simple cycle.
As is the case with any unary deterministic automaton, $\cB_q$ will have one run of length $n$, for every $n \in \NN$.
The cycle in $\cB_q$ will have the same length and overall weight as $\gamma_q$ but the accepting states and individual transition weights will likely differ.
The length of the simple path of $\cB_q$ will be bounded above by $2\abs{Q}^2$.

For every $q \in Q$, we define $\Pi_q$ to be the set of all runs in $\cA$ that (i) start at any initial state, (ii) end at any final state, and (iii) are of the form $\tau\gamma_q^x\tau'$ such that $\length(\tau) + \length(\tau') < 2\abs{Q}^2$ and $x \in \NN$.
It is important to note that $\Pi_q$ contains all runs of length less than $2\abs{Q}^2$, even those that do not visit $q$.
This fact will be most convenient for proving that $\cA$ can be converted into an equivalent quadratic size unary WFA in Chrobak normal form (in particular for~\cref{clm:same-paths} on~\cpageref{clm:same-paths}).

It is important to observe that runs across all $\Pi_q$ suffice for finding minimal weight runs.
Thanks to~\cref{lem:linear-form-runs}, it is true that for every $n \in \NN$ such that $\cA(n) \neq \infty$, there exists $q \in Q$ and $\pi \in \Pi_q$ such that $\weight(\pi) = \cA(n)$.
Now, for every $q \in Q$, we will define a weight function $f_q : \NN \to \bbZinf$ based on the paths in $\Pi_q$. 
Precisely, $f_q(n) \coloneqq \min\set{\weight(\pi) : \pi \in \Pi_q \text{ and } \length(\pi) = n}$.

\begin{clarep} \label{clm:correctness-of-union}
    For every $n \in \NN$, $\cA(n) = \min_{q \in Q}\set{f_q(n)}$.
\end{clarep}
\begin{claimproof}
    This claim follows from~\cref{lem:linear-form-runs} and the fact that $\bigcup_{q \in Q} \Pi_q$ contains all runs of the form $\tau\gamma_q^x\tau'$ that start from an initial state, and such that $\length(\tau)+\length(\tau') < 2\abs{Q}^2$ and $x \in \NN$.
\end{claimproof}

Our goal is to construct $\cB_q$ so that, for every $n \in \NN$, the run of length $n$ from the initial state of $\cB_q$ has weight $f_q(n)$.
Later, we will fix which control states are accepting to capture the situations when $\cA(n) \in \ZZ$ and $\cA(n) = \infty$.
Now, we will prove that $f_q(n)$ is periodic from $n \geq 2\abs{Q}^2$. 
We will use this fact to create the cycle of $\cB_q$.

\begin{claim} \label{clm:perioid-weights}
    Fix $q \in Q$.
    For every $n \geq 2\abs{Q}^2$, $f_q(n + \length(\gamma_q)) = f_q(n) + \weight(\gamma_q)$.
\end{claim}

\begin{claimproof}
    Without loss of generality, we shall assume that $\weight(\gamma_q) \geq 0$; the scenario in which $\weight(\gamma_q) < 0$ is analogous.

    First, we shall consider the case in which $f_q(n) \in \ZZ$ (i.e.\ $f_q(n) \neq \infty$).
    Given that $f_q(n) \in \ZZ$, there exists a path $\tau\gamma_q^x\tau' \in \Pi_q$ of length $n$ with weight $f_q(n)$.
    From this, one can immediately construct a run of length $n + \length(\gamma_q)$: $\tau\gamma_q^{x+1}\tau' \in \Pi_q$.
    This run has weight $f_q(n) + \weight(\gamma_q)$, thus $f_q(n + \length(\gamma_q)) \leq f_q(n) + \weight(\gamma_q)$.

    Now, for sake of contradiction, assume that $f_q(n + \length(\gamma_q)) < f_q(n) + \weight(\gamma_q)$.
    Suppose that $\tau\gamma_q^y\tau'$
    is the run in $\Pi_q$ of length $n + \length(\gamma_q)$ with minimal weight, $f_q(n+\length(\gamma_q))$.
    Since $n + \length(\gamma_q) \geq 2\abs{Q}^2 + \length(\gamma_q)$ and $\length(\tau) + \length(\tau') < 2\abs{Q}^2$, we know that $y \geq 1$.
    This means that $\tau\gamma_q^{y-1}\tau' \in \Pi_q$ is a run of length $n$ and weight $f_q(n+\length(\gamma_q)) - \weight(\gamma_q)$.
    Thus $f_q(n) \leq f_q(n+\length(\gamma_q)) - \weight(\gamma_q)$.
    This creates a contradiction as $f_q(n + \length(\gamma_q)) - \weight(\gamma_q) < f_q(n)$.
    Thus, in the case that $f_q(n) \in \ZZ$, $f_q(n+\length(\gamma_q)) = f_q(n) + \weight(\gamma_q)$. 

    Now, to conclude the proof, we shall consider the case in which $f_q(n) = \infty$. 
    In this case, there is no path of length $n$ in $\Pi_q$.
    Suppose, for sake of contradiction, that $f_q(n + \length(\gamma_q)) \neq \infty$. 
    This means that there is a path $\tau \gamma_q^y \tau' \in \Pi_q$ of length $n + \length(\gamma_q)$. 
    Since $n + \length(\gamma_q) \geq 2\abs{Q}^2 + \length(\gamma_q)$ and $\length(\tau) + \length(\tau') < 2\abs{Q}^2$, we know that $y \geq 1$.
    This means that $\tau \gamma^{y-1} \tau' \in \Pi_q$ is a run of length $n$ and weight $f_q(n + \length(\gamma_q) - \weight(\gamma_q) < \infty$.
    This creates a contradiction, and so we conclude that if $f_q(n) = \infty$, then $f_q(n + \length(\gamma_q)) = f_q(n) + \weight(\gamma_q) = \infty$. 
\end{claimproof}

Now we are ready to construct $\cB_q$.
There will be $2\abs{Q}^2 + \length(\gamma_q)$ many states: $2\abs{Q}^2-1$ will form a simple path leading to a cycle of $\length(\gamma_q)+1$ many states.
We use $\beta_q$ to refer to the cycle of $\cB_q$.

First, we shall consider which states of $\cB_q$ are accepting.
This part is straightforward: for every $n < 2\abs{Q}^2 + \length(\gamma_q)$ for which $f_q(n) \in \ZZ$, we mark the $n$-th state of $\cB_q$ as accepting. 
Otherwise, if $f_q(n) \neq \infty$, then the $n$-th state is not an accepting state.

Now, let $n_1 < n_2 < \ldots < n_\ell$ be the lengths of words accepted by $\cA$ such that (i) for every $i \in [\ell]$, $f_q(n_i) \in \ZZ$, (ii) $n_{\ell-1} < 2\abs{Q}^2 + \length(\gamma_q) \leq n_\ell$, and (iii) for every $n \notin \set{n_1, \ldots, n_\ell}$ such that $n < n_\ell$, $f_q(n) = \infty$.
In other words, $n_1, n_2, \ldots, n_{\ell-1}$ are the lengths of words with accepting runs in $\Pi_q$ that are at most $2\abs{Q}^2 + \length(\gamma_q)$, and $n_\ell$ is the length of the shortest word with a run in $\Pi_q$ whose length is at least $2\abs{Q}^2 + \length(\gamma_q)$.
Now, consider the sequence of weight differences of consecutive accepted lengths: $t = (f_q(n_1), f_q(n_2) - f_q(n_1), \ldots, f_q(n_\ell) - f_q(n_{\ell-1}))$.
This sequence satisfies the telescoping property that, for every $i \in [\ell]$, $\sum_{j=1}^i t[j] = f_q(n_1) + (f_q(n_2) - f_q(n_1)) + \ldots + (f_q(n_i) - f_q(n_{i-1})) = f_q(n_i)$.

We will use the sequence $t$ for the transition weights.
Overall, we wish to construct $\cB_q$ so that when the $n$-th state is reached, the current weight of the run is $f_q(n)$, assuming that $f_q(n) \neq \infty$. 
Concretely, before $n_i$-th state is reached, the weight of the run must be $f_q(n_i)$. 
To achieve this, on the first transition (the transition leading out from the initial state), the transition weight will be $t[1] = f_q(n_1)$. 
Then, for every $i \in [\ell]$, the transition from the $n_i$-th state has the weight $t[i+1]$. 
On transitions from non-accepting states, we just set the transition weight to $0$. 
See~\cref{fig:deterministic-U-WFA} for an example of $\cB_q$ in which $n_1 = 1, n_2 = 3, \ldots, n_i = 2\abs{Q}^2, \ldots, n_\ell = 2\abs{Q}^2+\length(\gamma_q)$.

\begin{figure}[t]
    \centering
\begin{tikzpicture}
    \node[state] (q00) at (-2,0) {};
    \node[state] (q0) at (-0.5,0) {};
        \node[state, inner sep = 0.6mm] at (-0.5,0) {};
    \node[state] (q1) at (1,0) {};
    \node[state] (q2) at (2.5,0) {};
        \node[state, inner sep = 0.6mm] at (2.5,0) {};
    \node (q3) at (4,0) {$\cdots$};
    \node[state] (q4) at (5,0) {};
    \node[state] (q5) at (6.5,0) {};
        \node[state, inner sep = 0.6mm] at (6.5,0) {};
    \node[state] (q6) at (8,0) {};
        \node[state, inner sep = 0.6mm] at (8,0) {};
    \node[state] (q7) at (9.5,0.75) {};
    \node[rotate=-10] (q8) at (11,0.75) {$\cdots$};
    \node[rotate=10] (q9) at (11,-0.75) {$\cdots$};
    \node[state] (q10) at (9.5,-0.75) {};

    \node at (10.15,0) {$\beta_q$}; 

    \draw[transition] (-2.5,0) -- (q00);

    \draw[transition] (q00) -- node[above, scale=0.8]{$t[1]$} (q0);
    \draw[transition] (q0) -- node[above, scale=0.8]{$t[2]$} (q1);
    \draw[transition] (q1) -- node[above, scale=0.8]{$0$} (q2);
    \draw[transition] (q2) -- node[above, scale=0.8,pos=0.4]{$t[3]$}(q3);
    \draw[transition] (q3) -- (q4);
    \draw[transition] (q4) -- node[above, scale=0.8]{$0$} (q5);
    \draw[transition] (q5) -- node[above, scale=0.8]{$t[i]$} (q6);

    \draw[transition] (q6) edge[out = 30, in = 210, bend left = 15] node[above, scale=0.8, sloped] {$t[i+1]$} (q7);
    \draw[transition] (q7) edge[out = 10, in = 180, bend left = 10] node[above, scale=0.8, sloped] {$0$} (q8);
    \draw[transition] (q9) edge[out = 10, in = 180, bend left = 10] (q10);
    \draw[transition] (q10) edge[out = 30, in = 210, bend left = 15] node[below, scale=0.8, sloped] {$t[\ell]$} (q6);
\end{tikzpicture}

     \vspace{-3mm}
    \caption{
        The deterministic unary WFA $\cB_q$.
    }
    \label{fig:deterministic-U-WFA}
\end{figure}
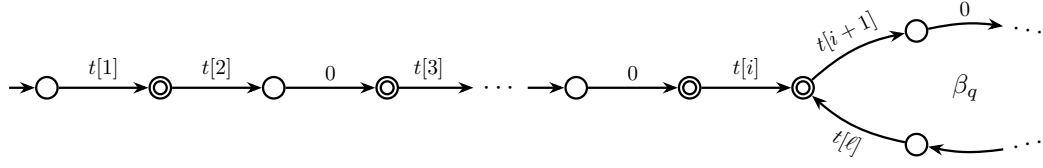

\begin{clarep} \label{clm:size-of-b}
    $\norm{\cB_q} \leq 8\abs{Q}^2\norm{\cA}$.
\end{clarep}
\begin{claimproof}
    For any $n \in \NN$, any run $\pi$ in $\cA$ of length $n$ must satisfy $\abs{\weight(\pi)} \leq n\norm{\cA}$.
    This means that $\abs{f_q(n)} \leq n\norm{\cA}$.
    Recall that the transition weights in $\cB_q$ are either $0$ or $t[i] = f_q(n_i) - f_q(n_i-1)$, for $1 \leq n_{i-1} < n_i < 2\abs{Q}^2 + 2\cdot\length(\gamma_q)$.
    Thus $\abs{t[i]} \leq n_i\norm{\cA} - (-n_{i-1}\norm{\cA}) = (n_{i-1}+n_i)\norm{\cA} \leq 2n_i\norm{\cA}$.
    The greatest value of $n_i$ is $2\abs{Q}^2 + 2\cdot\length(\gamma_q) - 1 \leq 2\abs{Q}^2 + 2\abs{Q} \leq 4\abs{Q}^2$. 
    Together, we conclude that the greatest transition weight $\norm{\cB_q}$ is bounded above by $8\abs{Q}^2\norm{\cA}$.
\end{claimproof}

\begin{clarep} \label{clm:cycle-weight}
    $\weight(\beta_q) = \weight(\gamma_q)$.
\end{clarep}
\begin{claimproof}
    Recall that $n_1, n_2, \ldots, n_{\ell-1}$ are the lengths of words with accepting runs in $\Pi_q$ that are at most $2\abs{Q}^2 + \length(\gamma_q)$, and $n_\ell$ is the length of the shortest word with a run in $\Pi_q$ whose length is at least $2\abs{Q}^2 + \length(\gamma_q)$. 
    Let $i$ is the least index such that $n_i \geq 2\abs{Q}^2$.
    We will now argue that $n_\ell = n_i + \length(\gamma_q)$. 
    First, as there is a run $\tau \gamma_q^x \tau'$ of length $n_i$ in $\Pi_q$, we know that there exists a run of length $n_i + \length(\gamma_q)$ in $\Pi_q$. 
    Clearly as $n_i \geq 2\abs{Q}^2$, we deduce that $n_i + \length(\gamma_q) \geq 2\abs{Q}^2 + \length(\gamma_q)$.
    This means that $n_\ell \leq n_i + \length(\gamma_q)$.
    Now, for sake of contradiction, assume that $n_\ell < n_i + \length(\gamma_q)$.
    Suppose $\tau \gamma_q^y \tau'$ is a run length $n_\ell$ in $\Pi_q$.
    Since $n_\ell \geq 2\abs{Q}^2 + \length(\gamma_q)$ and $\length(\tau) + \length(\tau') < 2\abs{Q}^2$, we deduce that $y \geq 1$.
    Thus, $\tau \gamma_q^{y-1} \tau' \in \Pi_q$ is a run of length $n_\ell - \length(\gamma_q)$. 
    Since $n_\ell \geq 2\abs{Q}^2 + \length(\gamma_q)$, it holds that $n_\ell - \length(\gamma_q) \geq 2\abs{Q}^2$. 
    Moreover, since $n_\ell < n_i + \length(\gamma_q)$, we deduce that $n_\ell - \length(\gamma_q) < n_i$. 
    Together, these facts contradict the definition of $n_i$.
    So we conclude that $n_\ell \geq n_i + \length(\gamma_q)$ and $n_\ell \leq n_i + \length(\gamma_q)$.

    Now, consider the transition weights of $\beta_q$. 
    Until the $n_i$-th state, the transition weights are 0.
    Then, the transition from the $n_i$-th state has weight $t[i+1] = f_q(n_{i+1}) - f_q(n_i)$.
    Following this, until the $n_{i+1}$-st state, the transition weights are again 0.
    Then, the transition from the $n_{i+1}$-st state has weight $t[i+2] = f_q(n_{i+2}) - f_q(n_{i+1})$.
    This continues; overall the weight of $\beta_q$ is
    \begin{equation*}
        t[i+1] + \cdots + t[\ell] = (f_q(n_{i+1}) - f_q(n_i)) + \ldots + (f_q(n_\ell) - f_q(n_{\ell-1})) = f_q(n_\ell) - f_q(n_i).
    \end{equation*}
    Now we can combine the fact that $n_i + \length(\gamma_q) = n_\ell$ with~\cref{clm:perioid-weights}.
    Since $n_i \geq 2\abs{Q}^2$, we know that $f_q(n_\ell) = f_q(n_i + \length(\gamma_q)) = f_q(n_i) + \weight(\gamma_q)$.
    Thus, we conclude that $\weight(\beta_q) = f_q(n_\ell) - f_q(n_i) = \weight(\gamma_q)$. 
\end{claimproof}

\begin{clarep} \label{clm:deterministic-weights}  
    For every $n \in \NN$, $\cB_q(n) = f_q(n)$.
\end{clarep}
\begin{claimproof}
    This proof is split into two cases: $n < 2\abs{Q}^2 + \length(\gamma_q)$ and $n \geq 2\abs{Q}^2 + \length(\gamma_q)$.
    
    \proofsubparagraph*{Case 1:} $n < 2\abs{Q}^2 + \length(\gamma_q)$.
    Here, the run in $\cB_q$ of length $n$ simply takes the first $n$ transitions and does not complete an iteration of the cycle $\beta_q$.
    
    Recall that the $n$-th state of $\cB_q$ is accepting if $f_q(n) \in \ZZ$.
    Thus, if $f_q(n) = \infty$, we know that the one run of length $n$ in $\cB_q$ leads to a non-accepting state, so $\cB_q(n) = f_q(n) = \infty$. 

    Now, suppose that $f_q(n) \in \ZZ$.
    The one run $\pi$ of length $n$ in $\cB_q$ leads to an accepting state. 
    Recall that $n_1 < n_2 < \ldots < n_{\ell-1}$ are the lengths of words such that $f_q(n_i) \in \ZZ$ and $n_{\ell-1} < 2\abs{Q}^2 + \length(\gamma_q)$.
    Suppose that $i \in [\ell]$ is the index such that $n_i = n$. 
    By definition of the weights on the transitions of $\cB_q$, we conclude that $\weight(\pi) = t[1] + t[2] + \ldots + t[i] = f_q(n_1) + (f_q(n_2) - f_q(n_1)) + \ldots + (f_q(n_i) - f_q(n_{i-1})) = f_q(n_i)$.
    Thus $\cB_q(n) = \weight(\pi) = f_q(n)$, as required.
    
    \proofsubparagraph*{Case 2:} 
    $n \geq 2\abs{Q}^2 + \length(\gamma_q)$.
    We will prove this case by induction on $n$.
    The base case is handled by Case 1.
    Assume $B_q(n') = f_q(n')$ for all $n' < n$. 
    Since $n \geq 2\abs{Q}^2 + \length(\gamma_q)$ and $\cB_q$ is a lasso, $\cB_q(n) = \cB_q(n-\length(\gamma_q)) + \weight(\beta_q)$. 
    By the inductive assumption, $\cB_q(n-\length(\gamma_q)) = f_q(n-\length(\gamma_q))$ and, thanks to~\cref{clm:cycle-weight}, we know that $\weight(\beta_q) = \weight(\gamma_q)$. 
    Thus, $\cB_q(n) = f_q(n-\length(\gamma_q)) + \weight(\gamma_q)$.
    We conclude this case thanks to~\cref{clm:perioid-weights} $f_q(n) = f_q(n-\length(\gamma_q)) + \weight(\gamma_q)$.
\end{claimproof}

We can now put the pieces together to prove~\cref{thm:deterministic-union}.
All that remains is to prove that $\cB_q$ can be computed from $\cA$ in polynomial time.

\begin{proof}[Proof of~\cref{thm:deterministic-union}]
    For every $q \in Q$, we construct $\cB_q$ as is described after~\cref{clm:perioid-weights}.
    Recall Claims~\ref{clm:correctness-of-union} and~\ref{clm:deterministic-weights}: for every $n \in \NN$, $\cA(n) = \min_{q \in Q}\set{f_q(n)}$ and $\cB_q = f_q(n)$, for every $q \in Q$.
    Together, we obtain condition~\ref{itm:union-correctness} of~\cref{thm:deterministic-union}.

    Condition~\ref{itm:union-cycle-count} follows from the fact that there is at most one deterministic unary WFA $\cB_q$ for every state $q\in Q$ of $\cA$.

    Condition~\ref{itm:union-state-count} immediately follows from the construction of $\cB_q$: there are at most $2\abs{Q}+\length(\gamma_q) \leq 2\abs{Q}^2 + \abs{Q}$ states and transitions (see~\cref{fig:deterministic-U-WFA} on~\cpageref{fig:deterministic-U-WFA}).
    Condition~\ref{itm:union-updates} is covered by~\cref{clm:size-of-b}.
    Finally, for condition~\ref{itm:union-poly-time}, it is important to note that we only need to know the length of the distinguished $q$-cycles; for this, we use~\cref{clm:compute-cycle-lengths}.
    We also need to compute the transition weights for $\cB_q$; for this, we use~\cref{clm:compute-fq}. 
    These two claims both rely on the basic fact (\cref{clm:compute-an}) that the weight of a minimal run of length $n$ can be computed in polynomial time from $\cA$ and $n$.
    We also note that~\cref{clm:compute-an} can be used to decide which of the states in $\cB_q$ are accepting in polynomial time.

    \begin{clarep} \label{clm:compute-an}
        Let $\cA'$ be an arbitrary unary WFA with states $Q'$,
        let $S, T \sset Q'$, and let $n \in \NN$.
        In polynomial time from $\cA'$ and $n$, we can compute the  weight of a minimal run that starts at some state $s \in S$, ends at some state $t \in T$, and has length $n$.
    \end{clarep}
    \begin{claimproof}
        Let $M \in \bbZinf^{\abs{Q'} \times \abs{Q'}}$ be the transition matrix of $\cA'$.
        In polynomial time with respect to $M$ and $n$, we can compute $M^n$.
        Now, we can simply compute $\min\set{M^n(s,t) : s \in S, t \in T}$.
    \end{claimproof}

    \begin{clarep}[Also follows from~\cite{karp1978characterization}]\label{clm:compute-cycle-lengths}
        Let $q \in Q$.
        We can compute, in polynomial time from $\cA$, the length of the $q$-cycle with the minimal average weight among all short $q$-cycles.
    \end{clarep}
    \begin{claimproof}
        Using~\cref{clm:compute-an} with $\cA' = \cA$, $S = \set{q}$, and $T = \set{q}$, we can compute the minimal weight $q$-cycle of length $n$ in polynomial time from $\cA$ and $n$.
        In particular, assuming $M$ is the transition matrix of $\cA$, it is true that $\frac{M^n(q,q)}{n}$ is the minimal average weight of a $q$-cycle of length $n$. 
        Since we are only concerned with short cycles, i.e. when $n \leq \abs{Q}$, we can identify the length of the $q$-cycle with minimal average weight among all short $q$-cycles by computing:
        \begin{equation*}
            \arg \min_{1 \leq n \leq \abs{Q}} \set{ \frac{M^n(q,q)}{n} }.
        \end{equation*}
    \end{claimproof}

    \begin{clarep} \label{clm:compute-fq}
        $f_q(n)$ can be computed in polynomial time with respect to $\cA$ and $n$.\qedhere
    \end{clarep}
    \begin{claimproof}
        There are two cases to consider: $n < 2\abs{Q}^2$ and $n \geq 2\abs{Q}^2$.

        \proofsubparagraph*{Case 1:} $n < 2\abs{Q}^2$.
        Since $\Pi_q$ contains all accepting runs of length less than $2\abs{Q}^2$ (regardless of $q$), we know that $f_q(n)$ is just the weight of a minimal accepting path of length $n$ (from any initial state, to any accepting state of $\cA$).
        For this, we use~\cref{clm:compute-an} with $\cA' = \cA$, $S = Q_0$, and $T = F$.

        \proofsubparagraph*{Case 2:} $n \geq 2\abs{Q}^2$.
        Recall that every run in $\Pi_q$ is of the form $\tau \gamma_q^x \tau'$ with $\length(\tau) + \length(\tau') < 2\abs{Q}^2$.
        That is, at most $2|Q|^2-1$ of the length of such runs do not come from iterations of~$\gamma_q$. 
        Our approach to compute $f_q$ is to construct a unary WFA $\cA'$ that tracks runs in $\Pi_q$ by verifying two conditions: (i) the length of the non-$\gamma_q$ part of the run is at most $2|Q|^2-1$ and (ii) $q$ is visited along the run. 

        \newcommand{\qpre}{\textup{pre}}
        \newcommand{\qpost}{\textup{post}}
        \newcommand{\qloop}{q_{\gamma,0}}
        Intuitively, $\cA'$ tracks in its states the following information:
        \begin{itemize}
            \item the current state of $\cA$,
            \item a counter ranging in $\{0,\ldots,2|Q|^2-1\}$ tracking the number of steps taken outside $\gamma_q$, and 
            \item an flag in $\{\qpre,\qpost\}$ indicating whether $q$ has already visited.
        \end{itemize}
        In addition, once an accepting state is reached with the $\qpost$ flag, $\cA'$ allows iterating $\gamma_q$ on a new designated accepting state. Note that this location is arbitrary; we could iterate this cycle at any other point, provided we ensure $q$ is visited.

        Thus, $\cA'$ assigns to $n$ the minimal value of a run of length $n$ in $\Pi_q$ that visits $q$, which (for runs of length at least $2|Q|^2$) is exactly $f_q(n)$.
        We proceed with the formal construction.

        The state space of $\cA'$ is $(Q\times \{0,1,\ldots,2|Q|^2-1\}\times\{\qpre,\qpost\})\cup \{q_{\gamma,0},\ldots,q_{\gamma,m}\}$. The initial states of $\cA'$ are $Q_0 \times \{0\} \times \set{\qpre}$.
        The transitions are as follows: for every transition $(s,c,t)\in \Delta$ of $\cA$ and for $0\le i<2|Q|^2-1$, we add the transitions $((s,i,\qpre),c,(t,i+1,\qpre))$ and $((s,i,\qpost),c,(t,i+1,\qpost))$ to $\cA'$. 
        Additionally, if $t \in F$ we also add the transition $((s,i,\qpost),c,\qloop)$ to $\cA'$.
        For transitions $(s,c,q) \in \Delta$, we add the transitions $((s,i,\qpre),c,(q,i+1,\qpost))$ and $((s,i,\qpre),c,\qloop)$ if $q\in F$. 

        In addition, the states $\{\qloop,\ldots,q_{\gamma,m}\}$ are a copy of the $\gamma_q$ cycle, including its weights.
        The only accepting state of $\cA'$ is $\qloop$.

        Thus, $\cA'$ proceeds for at most $2|Q|^2-1$ steps, transitioning from $\qpre$ to $\qpost$ upon visiting $q$, and after such a transition, reaching $\qloop$ if an accepting state is reachable. From there, the cycle $\gamma_q$ can be taken arbitrarily, and the run is accepting if and only if it traverses this cycle (in its entirety) any number of times.

        In particular, for every $n\ge 2|Q|^2$, $\cA'(n)$ is the minimal weight of a run that uses at most $2|Q|^2-1$ states outside $\gamma_q$, visits $q$ at least once, and reaches and accepting state, and possibly has any number of repetitions of $\gamma_q$. Hence $\cA'(n)=f_q(n)$.
        To conclude, since $\cA'$ can easily be constructed in polynomial time from $\cA$, we can again use~\cref{clm:compute-an} to compute $f_q(n)$ with $\cA' = \cA$, $S = Q_0 \times \set{0} \times \set{\qpre}$, and $T = \set{\qloop}$.
        \end{claimproof}
\end{proof}

\subsection{Extending Chrobak's Normal Form to Unary WFAs} 
\label{sec:cnf}

We will now strengthen~\cref{thm:deterministic-union} by proving that there exists a quadratic size unary WFA in Chrobak normal form that is equivalent to $\cA$ (which can be computed in polynomial time). 
To show this we consider the  the deterministic unary WFAs  $\cB_1, \ldots, \cB_k$ output by~\cref{thm:deterministic-union}, and observe that the stem of each $\cB_i$ is identical, meaning their stem can be merged, branching at the end of the stem to the cycle of the relevant $\cB_i$.
We state and prove this formally in the following corollary, for which we recall that we fixed our attention to a trim unary WFA $\cA = \tup{Q, Q_0, \Delta , F}$.

\begin{corollary}\label{cor:normal-form}
    There exists a non-deterministic unary WFA $\cB$ such that
    \begin{enumerate}[(1)]
        \item \label{itm:chrobak-correctness}
        for every $n$, $\cB(n) = \cA(n)$;
        \item \label{itm:chrobak-normal-form}
        the underlying unary NFA of $\cB$ is in Chrobak normal form;
        \item \label{itm:chrobak-size}
        $\cB$ consists of at most $3\abs{Q}^2$ states and transitions, in particular:
        \begin{enumerate}[(i)]
            \item \label{itm:chrobak-path-length}
            the initial path of $\cB$ has less than $2\abs{Q}^2$ states and
            \item \label{itm:chrobak-cycle-count}
            there are at most $\abs{Q}$ many cycles in $\cB$, each of length at most $\abs{Q}$;
        \end{enumerate}
        \item \label{itm:chrobak-updates}
        $\norm{\cB} \leq 8\abs{Q}^2\norm{\cA}$; and
        \item \label{itm:chrobak-poly-time}
        $\cB$ can be computed in polynomial time.
    \end{enumerate}
\end{corollary}

\begin{proof}
    Suppose $\cB_1, \ldots, \cB_k$ are the deterministic unary WFAs output by~\cref{thm:deterministic-union}. 
    Observe that in the construction, the initial paths (the first $2|Q|^2-1$ states and transitions) of every $\cB_i$ have the same weights and the same accepting states.
    This is formalised in the following claim and proof.
    
    \begin{claim}\label{clm:same-paths}
        For every $p, q \in Q$ and every $n < 2\abs{Q}^2$, $f_p(n) = f_q(n)$.
    \end{claim}
    \begin{claimproof}
        This follow from the fact that, for every $q \in Q$, $\Pi_q$ contains all runs of length less than $2\abs{Q}^2$ that start in any initial state, and that end in any final state.
        Thus, for every $p, q \in Q$, $\set{\pi \in \Pi_p : \length(\pi) < 2\abs{Q}^2} = \set{\pi \in \Pi_q : \length(\pi) < 2\abs{Q}^2}$.
        Hence $f_p(n) = f_q(n)$, for every $n < 2\abs{Q}^2$.
    \end{claimproof}
    
    To add some detail, recall that the transition updates in the initial path of $\cB_q$ are either $t[i] = f_q(n_i) - f_q(n_{i-1})$ or $0$ (depending only on whether the transition was outgoing from an accepting state).
    Given that the accepting states in the initial path of $\cB_q$ only depend on whether $\cA$ has an accepting run of length $n$ (for $n < 2\abs{Q}^2)$), thanks to~\cref{clm:same-paths}, it follows that the $n$-th transition of each $\cB_q$ has the same weight.  
    
    Accordingly, it is possible to merge the first $2\abs{Q}^2 - 1$ states and transitions of $\cB_1, \ldots, \cB_k$ to obtain $\cB$.
    After the $(2\abs{Q}^2-1)$-st transition, there will then be a non-deterministic branch to each of the cycles from $\cB_1, \ldots, \cB_k$.
    It is immediate that $\cB$ satisfies the structural conditions~\ref{itm:chrobak-normal-form},~\ref{itm:chrobak-path-length}, and~\ref{itm:chrobak-cycle-count} of~\cref{cor:normal-form}.
    The remaining conditions come directly from~\cref{thm:deterministic-union}. 
    The first condition follows from condition~\ref{itm:union-correctness} of~\cref{thm:deterministic-union}, the fourth follows from condition~\ref{itm:union-updates} of~\cref{thm:deterministic-union}, and the fifth follows from condition~\ref{itm:union-poly-time} of~\cref{thm:deterministic-union}.
\end{proof}

\section{Register Minimisation of CRAs}\label{sec:regmin}

\appendixlabel{appendix:regmin}
\begin{toappendix}
\label{appendix:regmin}
\end{toappendix}

Tropical cost register automata (CRA) are deterministic automata augmented with registers taking values in $\mathbb{Z}$, denoted by $k$-CRA when there are $k$ registers. Each transition specifies how the register values are updated linearly (within the tropical semiring) with respect to the other transitions. Intuitively, register updates take the following form
$r_i \xleftarrow{} \min\{r_1 + m_{1,i}, r_2 + m_{2,i},\ldots, r_k + m_{k,i}\}$, where $r_i's$ are registers and $m_{i,j}$ are constants. Note that since CRAs are deterministic, then over a unary alphabet they are lasso-shaped. The complexity of their behaviour stems from the register updates.
A formal definition of CRAs can be found in the appendix.

\begin{toappendix}
We start by formally defining CRAs. 
    A \emph{Cost Register Automaton with linear register updates} of dimension $k$ ($k$-CRA) is a tuple
$\cN = \tup{Q,\Sigma,\delta,q_0,F,\upd,\fin}$
with the following components:
\begin{itemize} 
    \item $\tup{Q,\Sigma,\delta,q_0,F}$ is a deterministic finite automaton (DFA), with $\delta:Q\times \Sigma\to Q$ a deterministic transition function.
    \item For each state $q\in Q$ and letter $\sigma\in \Sigma$, the \emph{update}
    $\upd(q,\sigma)$ is a $k\times k$ matrix over $\bbZinf$. We refer to its entries as $\upd(q,\sigma)_{i,j}$ for $1\le i,j\le k$.
    \item For each $q\in F$, the \emph{output registers} are $\fin(q)\subseteq [k]$.
\end{itemize}

We turn to define the semantics of CRAs.
Recall that all matrix products are in the $(\bbZinf,\min,+)$ semiring.  
A \emph{valuation} of the $k$ registers is a row vector $\vec{r}\in \bbZinf^k$.
Given such a valuation and a transition update $\upd(q,\sigma)$, the \emph{next valuation} is $\vec{r'}=\vec{r}\cdot \upd(q,\sigma)$.
Denote $M=\upd(q,\sigma)$ with entries $m_{i,j}$. Unfolding tropical matrix multiplication gives the following register update:
\[
    r_i' = \min\{r_1 + m_{1,i}, r_2 + m_{2,i},\ldots, r_k + m_{k,i}\}
\]

Consider a word $w=\sigma_1\cdots \sigma_n$ and the unique run
\[
    \rho=q_0\xrightarrow{\sigma_1}q_1\xrightarrow{\sigma_2}\cdots\xrightarrow{\sigma_n}q_n
\]
of the DFA part of $\cN$ on $w$.
Then $\rho$ induces a sequence of valuations $\vec{r^0},\ldots,\vec{r^n}$ given by the zero vector $\vec{r^0}=\vec{0}$, and for every $0\le i< n$ we have $\vec{r^{i+1}}=\vec{r^i}\cdot \upd(q_i,\sigma_{i+1})$. 
If $q_n\in F$, the value assigned by $\cN$ to $w$ is $\min\{\vec{r^n}(i)\mid i\in \fin(q_n)\}$.
If $q_n\notin F$ the value is $\infty$.
\end{toappendix}

A natural question is whether their representations can be simplified by reducing the number of registers. This leads to the notion of \emph{CRA Dimension Reduction}: given a $k$-CRA $\cA$ and an integer $k' < k$, does there exist an equivalent $k'$-CRA $\cB$? In this section we will show that the problem is in \coNP (and later, that it is \coNP-complete in \cref{thm:detcoNp-complete}). Note that in the general alphabet case, this problem is undecidable~\cite{almagor2025unambiguisability}.

To study this problem, we exploit the fundamental result that CRA are expressively equivalent to WFA~\cite{alur2013regular}. In one direction, the idea is that a WFA can non-deterministically track each register of a CRA using its states. For the converse, a CRA can simulate a non-deterministic WFA by keeping a register for each state of the WFA, and updating the registers according to the transitions of the WFA.

Crucially, this equivalence can be refined: the number of registers corresponds exactly to the width of the WFA. The following is proved in~\cite{almagor2025unambiguisability}.
\begin{theorem}[\cite{almagor2025unambiguisability}]
    \label{apx:thm:k CRA equivalent to k width WFA}
    The class of functions representable by a $k$-CRA is exactly that representable by width $k$ WFAs (and the translation between the two is effective).
\end{theorem}

By this theorem, CRA Dimension Reduction is equivalent to the \emph{WFA Width Reduction} problem which asks, given a $k$-width WFA $\cA$ and $k'<k$,
does there exist an equivalent $k'$-width WFA $\cB$? Both problems generalise the determinisation problem which can be stated as asking whether there is an equivalent $1$-width WFA or a $1$-CRA.

Since the two problems are equivalent, it suffices to work entirely in the WFA setting, which we do throughout to prove the following:

\begin{theorem}
    \label{thm:CRA minimisation is}
    The CRA Dimension Reduction  and the unary WFA Width Reduction problems are in \coNP.
\end{theorem}

\subparagraph*{Proof overview.}
To prove Theorem~\ref{thm:CRA minimisation is}, we develop a procedure
that computes the minimal width of a unary WFA \(\cA\).
Our procedure goes through the three following steps.
\begin{enumerate}
    \item
    First, we decompose \(\cA\) in polynomial time
    into a family of deterministic unary WFA \((\cB_i)_{i=1}^{k''}\).
    \item 
    Next, we group in polynomial time all components \(\cB_i\) that share the same asymptotic behaviours,
    which yields a smaller family of unary WFAs \((\cC_i)_{i=1}^{k'}\).
    \item
    Finally, we identify and discard \emph{redundant} components \(\cC_i\),
    namely those whose domain is eventually covered by other components with a smaller average cycle weight.
    We show that we can establish whether a component is redundant in \coNP (Proposition~\ref{prop:coNP}).
\end{enumerate}
We then prove that the number of non-redundant components obtained in the last step
is exactly the minimal width of any unary WFA equivalent to \(\cA\)
(Propositions~\ref{prop:high} and~\ref{prop:low}).

\subparagraph*{Complexity.} Note that the first two steps can be carried out in polynomial time,
and that redundancy of a given component \(\cC_i\) can be decided in \coNP,
that is, if \(\cC_i\) is not redundant, then it has a polynomial size non-redundancy witness. 

To solve the unary WFA Width Reduction problem,
it suffices to decide whether there exists
a \(k'\)-width WFA equivalent to \(\cA\),
without requiring minimality of \(k'\).
Our procedure yields a \coNP algorithm:
a certificate of non-existence of a \(k'\)-width unary WFA
consists of \(k'+1\) non-redundancy witnesses
for distinct components.

Let us now address the unary WFA Width Minimisation problem:
given a $k$-width unary WFA $\cA$ and an integer $k' < k$,
is \(k'\) the \emph{minimal} integer such that
there exists a $k'$-width unary WFA equivalent to \(\cA\)?
Formally, let \(\opt(\cA)\) be the minimal value
such that there exists an equivalent \(\opt(\cA)\)-width 
unary WFA.
Our procedure yields a \(\class{D^P}\) algorithm,
where \(\class{D^P}\) is the class of languages
that are the intersection of a language in \coNP
and a language in \NP
(see~\cite[Section 2]{PapadimitriouY84}).
Indeed, we can decide in \coNP
whether $k' \ge \opt(\cA)$,
and dually in \NP whether \(k' - 1 < \opt(\cA)\).
The intersection of these two conditions characterises
\(k' = \opt(\cA)\), which concludes the argument.

\subsection{The Procedure}
For the rest of the section we fix a unary WFA \(\cA\),
and we now detail each step of the procedure, before proving its correctness.

\subparagraph*{Step 1: Decomposition.}
We begin by applying Theorem~\ref{thm:deterministic-union}
to decompose \(\cA\) in polynomial time
into a family of deterministic unary WFAs \((\cB_i)_{i=1}^{k''}\) such that
\begin{equation}\label{equ:Bdecomposition}
    \cA(n) = \min_{1 \leq i \leq k''} \set{\cB_i(n)} \text{ for all } n \in \mathbb{N}.
\end{equation}

\subparagraph*{Step 2: Grouping components.}
We now group together the components \(\cB_i\)
that share the same asymptotic behaviour,
that is, the same average cycle weight.
Formally, let \(\alpha_1 < \alpha_2 < \cdots < \alpha_{k'}\)
be the distinct average weights of the cycles of the automata \((\cB_i)_{i=1}^{k''}\).
For each \(1 \leq i \leq k'\), we define \(\cC_i\) as the union
of all components \(\cB_j\) whose cycle has average weight \(\alpha_i\).
Note that the average cycle weights of the components \(\cB_j\) can be computed and compared in polynomial time,
therefore, as we do not determinise \(\cC_i\),
this construction can be done in polynomial time.
Since each \(\cB_j\) appears in one \(\cC_i\), Equation~\eqref{equ:Bdecomposition} yields
\begin{equation}\label{equ:Cdecomposition}
    \cA(n) = \min_{1 \leq i \leq k'} \set{\cC_i(n)} \text{ for all } n \in \mathbb{N}.
\end{equation}
Moreover, 
since all the cycles in \(\cC_i\) have average weight \(\alpha_i\), we  get the following:
\begin{lemrep}\label{lemma:asymptotic}
    There exists \(N \in \mathbb{N}\) such that, for all \(1 \leq i \leq k'\) and all \(n \in \mathbb{N}\),
    every run \(\rho\) of 
    \(\cC_i\) of length \(n\) satisfies
    \(|\weight(\rho) - n \cdot \alpha_i| \leq N\). 
    Moreover $N$ is polynomial in the description of the $\cC_i$.
\end{lemrep}
\begin{proof}
    Fix \(1 \leq i \leq k'\).
    Let \(S_i \in \mathbb{N}\) be the number of states of \(\cC_i\),
    and let \(M_i \in \mathbb{N}\) be the maximal absolute weight of a transition in \(\cC_i\).
    Let \(\rho\) be an accepting run of \(\cC_i\) on some input \(n \in \mathbb{N}\).
    Since \(\cC_i\) is a finite union of deterministic unary WFA,
    \(\rho\) can be decomposed into a simple path \(\rho_1\) of length \(n_1 \leq S_i\)
    followed with a loop \(\rho_2\) of length \(n_2\) obtained by iterating a simple cycle.
    Since each transition has weight at most \(M_i\),
    \(|\weight(\rho_1)| < n_1 \cdot M_i\).
    Furthermore, by definition of \(\cC_i\) the average weight of all cycles in \(\cC_i\) is \(\alpha_i\),
    hence \(\weight(\rho_2) = n_2 \cdot \alpha_i\).
    Therefore,
    \[
        |\weight(\rho) - n \cdot \alpha_i|
        = |\weight(\rho_1) + \weight(\rho_2)  - n \cdot \alpha_i|
        \leq |\weight(\rho_1)| + |\weight(\rho_2) - n \cdot \alpha_i|
        \leq n_1(M_i + \alpha_i).
    \]
    Since \(n_1 \leq S_i\), we get that \(|\weight(\rho) - n \cdot \alpha_i| \leq S_i(M_i + \alpha_i)\).
    Setting \(N = \max_{1 \leq i \leq k'} S_i(N_i + \alpha_i)\) then concludes the proof. Clearly $N$ is polynomial in $\max_{1 \leq i \leq k'} \norm{\cC_i}\cdot S_i$, as required.
\end{proof}

\subparagraph*{Step 3: Discarding redundant components.}
We now formalise which elements of the decomposition \((\cC_i)_{i=1}^{k'}\) are necessary.
A component \(\cC_i\) is called \emph{redundant}
if it is eventually dominated domain-wise by the previous components:
\[
\cC_i \text{ is redundant if } \dom(\cC_i) \setminus \bigcup\limits_{j=1}^{i-1} \dom(\cC_j) \text{ is finite}.
\]
We show that while the words witnessing the non-redundancy might be exponentially large,
we can decide their non-existence in \coNP:
\begin{proposition}\label{prop:coNP}
    We can establish in \coNP whether a given component is redundant. 
\end{proposition}
\begin{proof}
A component \(\cC_i\)
is redundant if and only if its domain differs only finitely
from the union of the domains of the preceding components.
Let \(\cA\) be the automaton recognising \(\dom(\cC_i)\),
and let \(B\) be the disjoint union of the automata recognising the domains of all
\(\cC_j\) with \(j<i\).
Then redundancy reduces to checking whether \(L(\cA) \setminus L(\cB)\) is finite.
This is known to be decidable in \coNP~\cite[Lemma 7.15]{ChistikovKMP22}.
\end{proof}

\subsection{Minimal Width Equals the Number of Non-redundant Components}

We now show that the number of non-redundant components is exactly
the minimal width of any unary WFA equivalent to \(\cA\).
Let \(i_1 <  \cdots < i_k\) be the indices of the non-redundant components \(\cC_i\).
We begin with an auxiliary lemma showing that \(\cA\)
is eventually determined by the union of its non-redundant components.

\begin{lemrep}\label{lemma:partition}
    There exist \(N \in \mathbb{N}\) and a partition
    \[
        \{n \in \dom(\cA) \mid n \geq N\} = D_1 \cup \cdots \cup D_k
    \]
    into \(k\) infinite regular languages such that
    \(\cA(n) = \cC_{i_j}(n)\) for all \(1 \leq j \leq k\) and all \(n \in D_j\).
\end{lemrep}
This follows from Lemma~\ref{lemma:asymptotic},
which shows that each component \(\cC_i\)
is asymptotically linear with slope \(\alpha_i\),
together with the fact that the components
are ordered by increasing values of \(\alpha_i\).
The full proof is deferred to the appendix.

\begin{proof}
For each \(1 \leq i \leq k'\), let
\[
    E_i = \dom(\cC_i) \setminus \bigcup_{j<i} \dom(\cC_j).
\]
By definition, \(\cC_i\) is redundant if and only if \(E_i\) is finite.
Let \(N_1 \in \mathbb{N}\) be larger than every element of each finite set \(E_i\).
By Lemma~\ref{lemma:asymptotic}, there exists \(N_2 \in \mathbb{N}\)
such that for all \(1 \leq i \leq k'\) and all \(n \in \dom(\cC_i)\) we have
\(|\cC_i(n) - n \cdot \alpha_i| \leq N_2\).
Let
\[
    N = \max\left(
        N_1,
        2N_2 \cdot \Big\lceil \max_{i<j} \frac{1}{\alpha_j - \alpha_i} \Big\rceil
    \right).
\]
Let \(i_1, \dots, i_k\) be the indices of the non-redundant components, and define
\[
    D_j = \{\, n \in E_{i_j} \mid n \geq N \,\}.
\]
By the choice of \(N_1\), the sets \(D_1,\dots,D_k\) form a partition of
\(\{n \in \dom(\cA) \mid n \geq N\}\), and each \(D_j\) is infinite and regular.
Fix \(1 \leq j \leq k\) and \(n \in D_j\). For all \(i < j\), we have \(n \notin \dom(\cC_i)\), hence \(\cC_i(n) = +\infty\).
Furthermore, for \(i > j\) we obtain
\[
\begin{aligned}
    \cC_i(n)
    &\geq n \cdot \alpha_i - N_2 \\
    &= n \cdot \alpha_j + n(\alpha_i - \alpha_j) - N_2 \\
    &\geq n \cdot \alpha_j + N(\alpha_i - \alpha_j) - N_2  \\
    &\geq n \cdot \alpha_j + N_2 \\
    &\geq \cC_j(n).
\end{aligned}
\]
Thus \(\cC_j(n) \le \cC_i(n)\) for all \(1 \leq i \leq k'\), and since \(\cA\) is equivalent to the union of the \(\cC_i\) (Equation~\ref{equ:Cdecomposition}), it follows that \(\cA(n) = \cC_j(n)\).
\end{proof}
We now construct an automaton of width \(k\) equivalent to \(\cA\).
The key step is to show that each non-redundant component can be determinised.
For this, we rely on a classical
characterisation of determinisability in terms of \emph{$B$-gap witnesses}.

\begin{definition} 
\label{def:b-gapwitness}
For a bound $B\in \bbN$, a \emph{$B$-gap witness} is a pair of words $x,y\in \bbN$ and a state $q\in Q$ such that upon reading $x$, we have $\minweight(Q_0\runsto{x}q)>\minweight(Q_0\runsto{x}Q)+B$, but $\minweight(Q_0\runsto{xy}F)=\minweight(Q_0\runsto{x}q\runsto{y}F)$. 
\end{definition}
The following characterisation is folklore; we use the formulation of
\cite[Theorem~2.1]{almagor2026determinization}.

\begin{theorem}[folklore, {\cite[Theorem 2.1]{almagor2026determinization}}]
\label{thm:abs:det iff gaps}
    $\cA$ is determinisable if and only if there is a $B\in \bbN$ such that there is no $B$-gap witness.
\end{theorem}
The determinisation construction that corresponds to the ``if'' direction of \cref{thm:abs:det iff gaps} yields a deterministic equivalent of size polynomial in $B$ (the maximal gap) and single-exponential in the number of states of $\cA$. 
We can now derive the upper bound.
\begin{proprep}\label{prop:high}
    There exists a \(k\)-width unary WFA equivalent to \(\cA\) of size single-exponential in $\cA$.
\end{proprep}

\begin{proof}
We first show that each non-redundant component \(\cC_{i_j}\) is determinisable.
By Lemma~\ref{lemma:asymptotic}, there exists \(N \in \mathbb{N}\)
such that for any two runs of \(\cA\) on the same input word,
the difference of their weights is bounded by \(2N\).
Hence no \(2N\)-gap witness exists,
and by Theorem~\ref{thm:abs:det iff gaps},
each \(\cC_{i_j}\) is determinisable.
Let \(\cD_j\) be a deterministic unary WFA equivalent to \(\cC_{i_j}\).

We now consider the unary WFA \(\cA'\) obtained
as the union of \(\cD_1,\dots,\cD_k\).
Since each \(\cD_i\) is deterministic,
\(\cA'\) has width \(k\).
Moreover, by Lemma~\ref{lemma:partition},
there exists \(N \in \mathbb{N}\)
such that \(\cA(n)=\cA'(n)\)
for all \(n \geq N\).
Thus, \(\cA\) and \(\cA'\) differ only on finitely many inputs.
Since the disagreement is finite, it can be fixed
while preserving the determinism of the components.
We modify a single component, say \(\cD_1\) into \(\cD_1'\),
by equipping its state space with a bounded counter
ranging over \(\{0,\dots,N\}\),
which records the length of the processed input up to \(N\).
For inputs of length at most \(N\),
 \(\cD_1'\) produces the corresponding value of \(\cA\).
For inputs of length greater than \(N\), \(\cD_1'\) coincides with \(\cD_1\),
so no further modification of \(\cD_1\) is needed.
This transformation preserves the determinism of \(\cD_1\),
hence the resulting automaton still has width \(k\) and is equivalent to~\(\cA\).

Finally, from \cref{lemma:asymptotic} we have that $N$ is polynomial in $\norm{\cA}$, i.e., either polynomial in the input (in case $\cA$ is represented in unary) or single-exponential, in case $\cA$ is represented in binary. Either way, the construction of \cref{thm:abs:det iff gaps} yields equivalent automata of size single-exponential in $\cA$.
\end{proof}
We now establish the matching lower bound, showing that no automaton of smaller width
can be equivalent to \(\cA\). The argument exploits the existence of the \(k\)
infinite regular domains \(D_1,\dots,D_k\) where each non-redundant component dominates according to its slope.
We show that these regions enforce \(k\) pairwise distinct
behaviours in any equivalent automaton, which requires width at least \(k\).

\begin{proprep}\label{prop:low}
    Every unary WFA \
    \(\cA'\) equivalent to \(\cA\) has width at least \(k\).
\end{proprep}
\begin{proof}
    From Lemma~\ref{lemma:partition}, there exist infinite regular languages
    \(D_1,\dots,D_k\) partitioning all sufficiently large inputs such that, on each \(D_j\),
    the behaviour of \(\cA\) coincides with that of the non-redundant component \(\cC_{i_j}\).
    Since each \(D_j\) is an infinite unary regular language, it is ultimately periodic:
    there exist \(\lambda_j, \mu_j \in \mathbb{N}\) such that
    \(n\lambda_j + \mu_j \in D_j\) for all \(n \in \mathbb{N}\).
    Let \(\lambda\) be the least common multiple of the \(\lambda_j\).
    By Lemma~\ref{lemma:asymptotic}, there exists \(N \in \mathbb{N}\) such that
    \[
        |\cA(n\lambda + \mu_j) - (n\lambda + \mu_j) \cdot \alpha_j|
        \leq N
    \text{ for all } 1 \leq j\leq k \text{ and } n \in \mathbb{N}.
    \]
    We now fix parameters for a separation argument. Let
    \[
    \mu = \max_{1 \leq j \leq k} \mu_j,\quad
    \alpha = \Big\lceil \max_{1 \leq i \leq k'} \alpha_i \Big\rceil,\quad
    \gamma = \Big\lceil \max_{1 \leq i < j \leq k'} \frac{1}{\alpha_j - \alpha_i} \Big\rceil,
    \]
    and let \(M'\) be the maximal absolute transition weight of \(\cA'\).
    Set
    \[
        x = 2 (N + \mu \cdot (M' + \alpha)) \cdot (\gamma + 1) \cdot \lambda.
    \]
    Since \(\cA'\) is equivalent to \(\cA\), for each \(j\) there is an accepting run
    \(\rho_j\) of \(\cA'\) on input \(x+\mu_j\) such that
    \[
    |\weight(\rho_j) - x \cdot \alpha_j|
    = |\cA(x+\mu_j)  - x \cdot \alpha_j|
    \leq |\cA(x+\mu_j)  - (x+\mu_j) \cdot \alpha_j| + \mu_j \cdot \alpha_j
    \leq N + \mu \cdot \alpha.
    \]
    Let \(\rho_j'\) be the prefix of \(\rho_j\) of length \(x\). Then
    \[
    |\weight(\rho_j') - x \cdot \alpha_j| \leq
    |\weight(\rho_j) - x \cdot \alpha_j| + \mu_j \cdot M'
    \leq N + \mu \cdot (M' + \alpha).
    \]
    Then for any \(1 \leq i < j \leq k\), since \(\alpha_i < \alpha_j\),
    the choice of \(x\) yields
    \[
    \begin{array}{lll}
    \weight(\rho_j') - \weight(\rho_i')
    & \geq &  x \cdot \alpha_j - (N + \mu \cdot (M' + \alpha)) - (x \cdot \alpha_i + N + \mu \cdot (M' + \alpha))\\
    & \geq &  x \cdot (\alpha_j-\alpha_i) - 2(N + \mu \cdot (M' + \alpha))\\
    & > &  2(N + \mu \cdot (M' + \alpha)) \cdot \lambda - 2(N + \mu \cdot (M' + \alpha))\\
    & \geq & 0.
    \end{array}
    \]
    Therefore, the values \(\weight(\rho_1'),\dots,\weight(\rho_k')\) are pairwise distinct,
    hence the runs \(\rho_1',\dots,\rho_k'\), all of length \(x\),
    must end in distinct states. It follows that \(\cA'\) has width at least \(k\).
\end{proof}

\section{Determinisation, Boundedness, and Hardness}

\appendixlabel{appendix:detboundhard}
\begin{toappendix}
\label{appendix:detboundhard}
\end{toappendix}

In this section, we consider the determinisation problem directly and show how it relates with boundedness problems. The determinisation problem is a special case of the register minimisation problem for the case of a single register (see~\cref{apx:thm:k CRA equivalent to k width WFA}).
Accordingly, our decision procedures and \coNP complexity upper bound carries over. 
We will now prove that the problems are both \coNP-hard. Furthermore, recall from~\cref{sec:chrobak} that a unary WFA can be represented as the union of $k$ deterministic automata.
We provide strong evidence that determinisation and boundedness are not fixed parameter tractable when $k$ is the parameter.

First, we show that determinisation is effectively interreducible with value boundedness. 
The \emph{value boundedness} problem is a specific variant of the boundedness problem which requires a sequence of words with increasing value (a single non-accepting word, with weight $\infty$, does not suffice).
This can be seen as the core of the weighted problem as the detection of non-accepting words is purely a reachability criterion.

\begin{definition}\label{defn:boundedness}
A weighted automaton $\cA$ is \emph{bounded} if there exists $B\in \mathbb{N}$ such that  $|\cA(n)| < B$ and \emph{not bounded} otherwise. 
Moreover, $\cA$ is \emph{value bounded} if there exists $B\in \mathbb{N}$ such that if $\cA(n) < \infty$, then $|\cA(n)| < B$, and \emph{not value bounded} otherwise\footnote{We avoid the term \emph{unbounded} to prevent ambiguity.}. 

The boundedness problem asks, given a weighted automaton $\cA$, whether $\cA$ is bounded, and the value boundedness problem asks whether $\cA$ is value bounded. 
\end{definition}

We show that the value boundedness problem is interreducible with the determinisation problem for unary weighted automata. 
Our approach is to show that (i) the connection is almost immediate if the minimal average cycle has weight zero, and (ii) there are reductions that adjust the systems to have minimal average cycle of zero weight suited to each direction of the equivalence. 
As a consequence we will show the following theorem.

\begin{restatable}{theorem}{thmeverythingconpcomplete}
\label{thm:detcoNp-complete}
    The determinisation, boundedness, value boundedness, register reduction, and width reduction problems for unary weighted automata are \coNP-complete.
\end{restatable}

\subsection{Automata with Zero-Weight Minimal Average Cycles}

In this section, we will establish the tight connection between value boundedness and determinisation in the case where the minimal average cycle has weight zero. 
We recall \cref{def:b-gapwitness} and \cref{thm:abs:det iff gaps} that tell us that $\cA$ is determinisable if and only if there is no $B$-gap witness for some $B$. 
When the minimal average slope has weight zero, bounding $B$-gap witnesses, and thus determinisability, it turns out to correspond exactly to the value boundedness problem. 

Recall that we assume $\cA$ is trim, so that any run visiting a reachable cycle can be extended into an accepting word. Given a weighted automaton $\cA$, a reachable cycle $\rho$ is of minimal average weight if $\weight(\rho)/\length(\rho) \le \weight(\rho')/\length(\rho')$, for all reachable cycles $\rho'$. Let $\alpha = \weight(\rho)/\length(\rho)$ be the minimal average slope achieved by a cycle in $\cA$.\footnote{\cref{sec:chrobak} considered minimal average (short) $q$-cycles. Here the property is global to the automaton.}

\begin{clarep}[\cite{karp1978characterization}]
    A cycle achieving the minimal average slope $\alpha$ is simple, and therefore has of length less than $|Q|$.
    Moreover, $\alpha$ can be computed in polynomial time.
\end{clarep}
\begin{claimproof}
    The claim is proved in~\cite{karp1978characterization}. We bring its proof for completeness.
    Let $\rho$ be the shortest cycle achieving the minimal average slope, and suppose it starts and ends at $q \in Q$. 
    Suppose for contradiction some state $p\in Q$ in the cycle repeats (other than $q$ at the start and end).
    Consider the cycle $\rho'\subset \rho$ between two occurrences of $p$. 
    If $\rho'$ has smaller average weight than $\rho$, this contradicts the fact that $\rho$ has minimal average weight. 
    If $\rho'$ has the same average weight then $\rho$ can be made shorter by removing $\rho'$, maintaining the same average weight, contradicting that $\rho$ is the shortest. If $\rho'$ has greater average weight then $\rho'$ can also be removed, making the overall average weight lower, contradicting that $\rho$ achieved the minimal average slope. Hence, $\rho$ is simple, and thus has length less than $|Q|$. 
    
    Suppose $M$ is the transition matrix of $\cA$ and observe that $\alpha = \min_{1\le n\le |Q|, q\in Q} \left\{\frac{M^n(q,q)}{n}\right\}$. 
    Just like in the proof of~\cref{clm:compute-cycle-lengths}, we can compute each of these values, of which there are polynomially many, in polynomial time to determine the minimum.
\end{claimproof}

\begin{clarep}
\label{claim:lowerbandperiodicupperb}
Let $\cA'$ be an automaton with minimal average slope of zero.
There exists $B_0\in \bbN$ such that, for all $x\in \bbN$, $\cA'(x) \ge -B_0$  and $\ell,k\le2|Q|$ such that  $\cA'(\ell +kn)\le B_0$ for all $n\in \bbN$.
\end{clarep}
\begin{claimproof}
We attain the lower bound since no cycle in $\cA'$ can have negative weight. For the upper bound, let $k \le |Q|$ be the length of the simple cycle attaining the minimal average weight ($r/s$) in $\cA'$, and let $q$ be a state on the cycle. Consider the shortest run of the form  $Q_0\xrightarrow{\ell_1}q\xrightarrow{\ell_2} F$ of length $\ell = \ell_1+\ell_2 \le  2|Q|$. Then $\minweight(Q_0\xrightarrow{\ell+kn}F)\le\minweight(Q_0\xrightarrow{\ell_1}q \xrightarrow{\ell_2}F)\le B_0$
\end{claimproof}

\begin{claim}\label{claim:valbdet}
Let $\cA'$ be an automaton with minimal average slope of zero. 
If $\cA'$ is value bounded by $M$, then there is a $B$ such that there is no $B$-gap witness, so $\cA'$ is determinisiable.
\end{claim}
\begin{claimproof}
We note that $\cA'$ does not have negative weight cycles, thus, if the weight of the run exceeds $M_{\cA'}= M+|Q| \cdot||\cA'||$ then it can never be extended to a minimum weight run (which has value less than $M$). Suppose we have a $B$-gap witness, $\minweight(Q_0\runsto{x}q)>\minweight(Q_0\runsto{x}Q)+B$, but $\minweight(Q_0\runsto{xy}F)=\minweight(Q_0\runsto{x}q\runsto{y}F) \le M_{\cA'}$.  But we have $Q_0\runsto{x}q$ used in a minimal weight run, and thus $\minweight(Q_0\runsto{x}q) \le M_{\cA'}$, By \cref{claim:lowerbandperiodicupperb}, we have $\minweight(Q_0\runsto{x}Q) \ge - B_0$ for a constant $B_0\in\bbN$.
In summary, we have 
\[
M_{\cA'} \ge \minweight(Q_0\runsto{x}q)>\minweight(Q_0\runsto{x}Q)+B \ge - B_0 + B,
\]
and so all $B$-gap witnesses have $B \le M_{\cA'} + B_0$.
\end{claimproof}

\begin{claim}\label{claim:notvalbnotdet}
Let $\cA'$ be an automaton with minimal average slope of zero.
If $\cA'$ is not value bounded then there is a $B$-gap witness for all $B$, and so $\cA'$ is not determinisiable.
\end{claim}
\begin{claimproof}
Suppose $\cA'$ is not value bounded; we will construct arbitrary large $B$-gap witnesses. We have sequences $m_1,m_2,\ldots\in \bbN$, $L_1,L_2,\ldots\in\bbN$ such that $\cA'(m_i) = L_i$ and $L_i\xrightarrow[i\to\infty]{}\infty$. Consider a run such that $\minweight(Q_0 \xrightarrow{m_i} F) = L_i$. Consider the largest $m$ such that $\ell+km < m_i$. Let $q$ be the state reached after $\ell+km$ steps on the run attaining weight $L_i$. Since $m_i-\ell-km <|Q|$ we have $\minweight(Q_0\xrightarrow{\ell+km} q) \ge L_i - |Q|\cdot||\cA'||$. 
On the other hand we have, by \cref{claim:lowerbandperiodicupperb},
we have $\minweight(Q_0 \xrightarrow{\ell+km} Q) < B_0$ for a constant $B_0\in\bbN$. Hence, we have a $(L_i - |Q|\cdot||\cA'|| - B_0)$-gap witness formed by $x = \ell+kn$ and $y= m_i - \ell- kn$. Since $|Q|,||\cA'||, B_0$ are bounded  and $L_i\to\infty$, we have unbounded gap witnesses.
\end{claimproof}

\subsection{Union of $k$ Deterministic Automata}

We observed in \cref{sec:chrobak} that unary weighted automata can be represented as the union of deterministic automata, and can efficiently be transformed into this representation (\cref{thm:deterministic-union}). 
It is natural to consider how the number of such components affects the complexity.

In this section, we consider the boundedness, value boundedness and determinisability problems when the input is a union of $k$ deterministic automata. Recall that a parameterised
problem is fixed parameter tractable (FPT) if there exists an algorithm that solves any
instance $(I, k)$ of size $n$ in time $f(k)\cdot n^{O(1)}$, where $f$ is a computable function. In the sequel, we call the parametrised versions of boundedness, value boundedness, and determinisability $k$-det boundedness, $k$-det value boundedness, and $k$-det determinisability, respectively. 
We will prove that all three problems are not fixed parameter tractable.

Concretely, we show that all these problems are \cowone-hard when parameterised by the number of deterministic automata in the union. 
Based on standard assumptions in parameterised complexity, this provides strong evidence that the problems are not FPT~\cite{FlumG06}.

\begin{lemma}\label{lem:w1}
    $k$-det boundedness, $k$-det value boundedness and $k$-det determinisability are \cowone-hard.
\end{lemma}
\begin{proof}
    In~\cite{intersection_w1} it was shown that problem of nonnemptiness of $k$ deterministic unary automata is \wone-hard. Let us fix an instance of this problem, namely $k$ deterministic automata $(\cA_1, \cA_2, \ldots, \cA_k)$. Firstly, notice that $\bigcap_{i \in [k]} L(\cA_i) \neq \emptyset$ if and only if $\bigcup L(\overline{A_i}) \neq \set{a}^*$, where $\overline{A_i}$ is the automaton recognising the complement of $L(A_i)$ and can be computed in time polynomial in the size of $A_i$. Notice also that $\overline{A_i}$ is still a deterministic automaton. 
    Hence, the universality problem parameterised by the number of deterministic automata in the union is \cowone-hard.  

    \subparagraph*{$k$-det boundedness.}
    Let us fix an instance of the universality problem (parametrised by the number of automata in the union), namely $k$ deterministic unary automata $(\cA_1, \cA_2, \ldots, \cA_k)$. Let $\cA$ be the automaton, which is the union of $\cA_1, \cA_2, \ldots, \cA_k$. For each $i \in [k]$, we construct a deterministic unary WFA $\cB_i$. Concretely, the states of $\cB_i$ will be the states of $\cA_i$. Also, accepting and initial states of $\cA_i$ will be accepting and initial states of $\cB_i$, respectively. Finally, we set the weight of every transition in $\cB_i$ to zero.
    Let $\cB$ be the automaton, which is the union of $\cB_1, \cB_2, \ldots, \cB_k$. 
    Notice that if $\cA$ is universal, then for every $n \in \NN$, we have $\cB(n) = 0$. 
    On the other hand, if $\cA$ does not accept $a^n$, then we have $\cB(n) = \infty$. Since, for every $i \in [k]$, $\cB_i$ can be constructed in time polynomial in size of $\cA_i$ we conclude that $k$-det~boundedness is~\cowone-hard.

    \subparagraph*{$k$-det value boundedness.}
    Once again, let us fix as instance of the universality problem, namely $k$ deterministic unary automata $(\cA_1, \cA_2, \ldots, \cA_k)$. Let $\cA$ be the automaton, which is the union of $\cA_1, \cA_2, \ldots, \cA_k$ and let $U = \{n \mid a^n \notin L(\cA)\}$. 
    
    Notice, that if $U$ is finite, but nonempty, then $U \subseteq [0,n]$ where $n$ is the number of states of $\cA$. This is because if there exists integer $\ell > n$ such that $\ell \in U$ then, for every $i \in [k]$, automaton $\cA_i$ either cannot read $a^\ell$ or after reading $a^\ell$ is in a state $q$ which lies on a simple cycle. 
    Thus, also $\ell + n! \in U$, implying that $U$ is infinite, as for every $i \in [k]$, if automaton $\cA_i$  after reading $a^\ell$ is in a state $q$, then after reading $a^{\ell+n!}$ is also in the state $q$ as it can only repeat the simple cycle in which $q$ lies and the length of this cycle divides $n!$. 
    Observe that in polynomial time in $n$, one can check whether there is $\ell \in [0, n] \cap U$ and, if there exists such $\ell$, output a value unbounded single state, one transition (with weight one) unary WFA. 
    
    Otherwise, we can assume that $U$, if nonempty, is infinite and, for each $i \in [k]$, we construct a deterministic unary WFA $\cB_i$. Concretely, the states of $\cB_i$ will be the states of $\cA_i$. Also, accepting and initial states of $\cA_i$ will be accepting and initial states of $\cB_i$, respectively. Finally, we set the weight of every transition in $\cB_i$ to zero. We also construct an unary WFA $\cB_{k+1}$ with a single state $q$, which is initial and accepting, and a single transition $(q,1,q)$. Clearly, for every $i \in [k]$, $\cB_i$ is a deterministic unary WFA and can be constructed in time polynomial in size of $\cA_i$. Moreover, also $\cB_{k+1}$ is a deterministic unary WFA and can be constructed in constant time. 
    Hence, in order to show that $k$-det value boundedness is \cowone-hard, it suffices to show that automaton $\cA$ is universal if and only if $\cB$, the union of $\cB_1, \ldots, \cB_k$, is bounded. 
    
    Observe that if $n \in U$ then, for every $i \in [k]$, we have that $\cB_i(n) = \infty$. 
    Moreover, we have $\cB_{k+1}(n) = n$. Hence, if $\cA$ is not universal, meaning that $U$ is nonempty, and hence infinite we can conclude that $\cB$ is not value bounded. On the other hand, if $\cB$ is value bounded then $U$ needs to be finite, and hence empty, which means that $\cA$ is universal. This concludes the proof that $k$-det value boundedness is \cowone-hard.

    \subparagraph*{$k$-det determinisiability.}
    To observe that $k$-det determinisability is also \cowone-hard, we reduce from $k$-det value boundedness to $k$-det determinisability. 
    Concretely, having $k$ deterministic unary WFAs $(\cA_1, \cA_2, \ldots, \cA_k)$, for every $i \in [k]$, we construct a unary WFA $\cB_i$ from $\cA_i$ by replacing each transition $t  = (q, c, q')$ by two transitions $(q, c, s_{t})$ and $(s_{t}, 0, q')$ using new non-accepting state $s_t$. 
    Observe that, w.l.o.g.\ we can assume that, for every transition, $c \geq 0$, as the previous reduction produced a unary WFA only with such weights. 
    Moreover, we construct $\cB_{k+1}$ which is a zero-weight two-state component, alternating between an accepting and a non-accepting state, and with the initial state being the non-accepting state. 
    Clearly, for every $i \in [k]$, $\cB_i$ is a deterministic unary WFA and can be constructed in time polynomial in size of $\cA_i$. 
    Moreover, $\cB_{k+1}$ is deterministic and can be constructed in constant time. 
    Hence, in order to show that $k$-det~determinisibality is \cowone-hard it is enough to show that automaton $\cA$, that is a union of $\cA_1, \ldots, \cA_k$, is value bounded if and only if unary WFA $\cB$, that is a union of $\cB_1, \ldots, \cB_k$, is determinisable. 
    Firstly, notice that because of $\cB_{k+1}$ we have, for every $n \in \NN$, $B(2n+1) = 0$. Moreover, for every $i \in [k]$ and $n \in \NN$, we have $\cB_i(2n) = \cA_i(n)$.  Hence, $\cA$ is value bounded if and only if $\cB$ is value bounded. Due to $\cB_{k+1}$, the minimal average slope of $\cB$ is zero. Thus, by \cref{claim:valbdet} and \cref{claim:notvalbnotdet}, we have $\cB$ is determinisable if and only if $\cB$ is value bounded.
\end{proof}

\subsection{Determinisability and Value Boundedness}
In the preceding subsection, we showed that universality of $k$~DFAs, and value boundedness of $k$~unary WFAs reduce to the determinisation problem. 
We can observe that the reduction from value boundedness to determinisation applies with only minor modification, without assuming the input is a union of deterministic and that the minimal average slope is non-negative.

\begin{lemrep}
Boundedness and value boundedness reduce to determinisability in polynomial time.
\label{lemma:valboundtodet}
\end{lemrep}
\begin{proof}
Given $\cA$ we decide if is is value bounded using determinisation. Let $\alpha$ be the minimal average slope of $\cA$. We assume without loss of generality that $\alpha \ge 0$. If $\alpha = \alpha_q < 0$, computable in polynomial time, there is a decreasing sequence formed by pumping around $q$, since it is not value bounded we output any small non-determinisable unary WFA.

We construct the automaton $\cA'$, such that $\cA'(2n+1) = 0, \cA'(2n) = \cA(n)$. Each transition $q\xrightarrow{}q'$ is replaced by $q\xrightarrow{}s_{q'}\xrightarrow{}q$ using new non-accepting states $s_q$, where $\weight(q\to q') = \weight(q\to s_{q'})$ and $\weight(s_q \to q') = 0$. $\cA'(2n+1) = 0$ is easily obtained with a zero-weight two state component, alternating between accepting and non-accepting.  Clearly, $\cA$ is value bounded if and only if $\cA'$ is value bounded.  Due to the second component, we have $\cA'$ has minimal average slope of zero. Thus, by \cref{claim:valbdet} and \cref{claim:notvalbnotdet}, we have $\cA'$ determinisable iff $\cA'$ is value bounded.

It remains to show that boundedness reduces to determinisability, which we show wit the following claim:
\begin{claim} \label{lemma:boundedtovalbounded}
    Boundedness reduces to value bounded in polynomial time.
\end{claim}
\begin{claimproof}
Let $U = \{n \mid \cA(n) = \infty \}$. $\cA$ is bounded if $U$ is empty and $\cA$ is value bounded. $U$ is either empty, finite or infnite.

Consider the case $U \ne \emptyset$, either $U$ is finite or infinite. The set $U$ is recognised by an NFA formed from $\cA$ dropping the weights. The Chrobak normal form of this automaton tells us that the system is a short lasso followed by periodic behaviour. Hence if $\cA(n) = \infty$ for some $n \ge |Q|^2$ then $U$ is infinite. 
We can, in polynomial time check the first $|Q|^2$ values, and return ``not bounded'' if any have $\cA(n) = \infty$.

Henceforth, we can assume that either $U$ is empty or infinite. Let $\cA'(n) = \min\{\cA(n),n\}$, constructed from $\cA$ by a new initial and accepting state counting up by $1$.
We show that under the assumption $U$ is infnite or empty,  $\cA$ is bounded if and only if $\cA'$ is value bounded.
\begin{itemize}
    \item $\cA$ is bounded by $M$, then $\cA'(n) = \min\{\cA(n),n\} \le M$, and so $\cA'$ is both bounded and value bounded.
    \item $\cA$ is not bounded, either $U$ is empty or $U$ infinite:
    \begin{itemize}
        \item Suppose $U$ is infinite. For $n\in U$, $\cA'(n) = n$, and so $\cA'$ is not value bounded. 
        \item Suppose $U$ is empty, then $\cA$ is already not value bounded. Since the witness of increasing values must have increasing length, the automaton $\cA'$ is also not value bounded on the same sequence. \qedhere
    \end{itemize}
\end{itemize}
\end{claimproof}
    Hence, boundedness reduces to value boundedness, which reduces to determinisability.
\end{proof}

We can also show that the connection is tight for the value boundedness problem.

\begin{lemma} \label{lemma:dettovalueboundedness}
    Determinisability reduces to value boundedness in polynomial time.
\end{lemma}
    
\begin{proof}
Given a weighted automaton $\cA$, let $\alpha = r/s$, $s\le |Q|$, $gcd(r,s)=1$ be the minimum average slope in $\cA$, we define $\cA'(n) = s\cdot\cA(n) - r \cdot n$. 
This is obtained by multiplying the weight of every transition by $s$, and subtracting $r$. 

Since $\cA'$ has minimal average slope of zero, we have $\cA'$ determinisable iff $\cA'$ is value bounded by \cref{claim:valbdet} and \cref{claim:notvalbnotdet}. The proof is completed by the following claim that $\cA$ and $\cA'$ are equideterminisable, and so it suffices to determine whether $\cA'$ is value bounded.

\begin{clarep}
    $\cA$ is determinisable if and only if $\cA'$ is determinisable.
\end{clarep}
\begin{claimproof}
If $\cA$ is determinisable, realised by $\cB$, as the translation does not affect determinism, apply the same translation to $\cB$ to attain a deterministic version of $\cA'$.

If $\cA'$ is determinisable, realised by $\cB'$ we apply the converse operation to $\cB'$ attaining a deterministic automaton $\cA'\equiv \cA$. However, it is not certain that the transition weights in $\bbZ$, since we divide by $s$. However, we have that $\cA'$ is a lasso and has integer output weights (since the equivalent $\cA$ has integer output weights), it is therefore easy to adjust the automaton to integer weights. Indeed, if $q_0\xrightarrow{n}q\xrightarrow{}q'$ and $q,q'\in F$, we have $\weight(q\to q') = \cA(n+1)- \cA(n) \in \bbZ$. If $q_0\xrightarrow{n}q_1,q_2,\dots,q_m$, with $q_1,q_m\in F$ and $q_2,\dots,q_{m-1}\not\in F$ then we have $\cA(n+m)- \cA(n) \in \bbZ$, and we can update $\weight(q_1\xrightarrow{}q_2) = \cA(n+1)- \cA(n)$ and $q_i\xrightarrow{}q_{i+1})=0$ for $i=2,\dots,m-1$.
\end{claimproof}
\vspace{-1.5em}
\end{proof}
\begin{remark}[Remark on \cref{lemma:dettovalueboundedness}] 
    We note that we do not show that determinisability of a unary WFA $\cA$ reduces to boundedness (only to value boundedness), unless $\cA$ already happens to have a universal domain. Words not in the domain are no issue for determinisability, so it seems we would need to repair these holes by assigning them a bounded value, but we would need to do this without affecting the weight of other words: the natural approach seems to require the addition of an automaton for exactly $\overline{\dom(\cA)}$ (with weights zero for example); this automaton can have exponential size.
\end{remark}

\subsection{Proving \cref{thm:detcoNp-complete}}
Through a chain of reductions, we have observed that NFA universality reduces to boundedness (\cref{lem:w1}), reduces to value boundedness (\cref{lemma:valboundtodet}), reduces to determinisation (\cref{lemma:valboundtodet}), which is a special case of CRA dimension reduction/WFA width reduction. The former, NFA universality is \coNP-hard~\cite{intersection_w1} while \cref{thm:CRA minimisation is} shows the latter are in \coNP. 
\thmeverythingconpcomplete*

\section{Future Directions}
This paper considers, given unary WFA or their equivalent Cost Register Automata with linear register updates, how to represent them as simply as possible (as a union of deterministic automata, in Chrobak normal form, with lower width, or as CRA using fewer registers). 

In general, CRA can be defined over any update function that, in particular, need not be linear~\cite{alur2013regular}. 
The case of polynomial updates over the tropical semiring is particularly interesting, here the constant exponent essentially provides scalar multiplication alongside minimum and addition. Such models, which are considered in~\cite{HeerdtHOS25}, generalise both tropical and rational unary weighted automata, the latter often studied as Linear Recurrence Sequences. Despite their apparent simplicity, LRS have interesting algorithmic problems, particularly the Skolem and Positivity problems which remain wide open~\cite{OuaknineW15}. A natural direction is to consider boundedness, register minimisation, and normal forms for these tropical CRA with polynomial updates.

\end{document}